\documentclass[pdflatex,sn-mathphys-num]{sn-jnl}

\usepackage{graphicx}%
\usepackage{multirow}%
\usepackage{amsmath,amssymb,amsfonts}%
\usepackage{amsthm}%
\usepackage{mathrsfs}%
\usepackage[title]{appendix}%
\usepackage{xcolor}%
\usepackage{textcomp}%
\usepackage{manyfoot}%
\usepackage{booktabs}%
\usepackage{algorithm}%
\usepackage{algorithmicx}%
\usepackage{algpseudocode}%
\usepackage{listings}%
\usepackage{nicefrac,xfrac}%
\usepackage{dsfont}%
\usepackage{subcaption}%
\usepackage{tikz}%
\usetikzlibrary{arrows.meta,positioning,calc}%

\theoremstyle{thmstyleone}%
\newtheorem{theorem}{Theorem}
\newtheorem{corollary}{Corollary}
\newtheorem{lemma}{Lemma}
\newtheorem{identity}{Identity}
\newtheorem{proposition}[theorem]{Proposition}%

\theoremstyle{thmstyletwo}%
\newtheorem{example}{Example}%
\newtheorem{remark}{Remark}%

\newtheorem{specialcase}{Special Case}

\theoremstyle{thmstylethree}%
\newtheorem{definition}{Definition}%

\begin{document}

\title[Discrete Radar based on Modulo Arithmetic]{Discrete Radar based on Modulo Arithmetic}


\author*[1]{\fnm{Nishant} \sur{Mehrotra}}\email{nishant.mehrotra@duke.edu}
\equalcont{These authors contributed equally to this work.}

\author[1]{\fnm{Sandesh Rao} \sur{Mattu}}\email{sandesh.mattu@duke.edu}
\equalcont{These authors contributed equally to this work.}

\author[2]{\fnm{Saif Khan} \sur{Mohammed}}\email{saifkmohammed@gmail.com}

\author[3,4]{\fnm{Ronny} \sur{Hadani}}\email{hadani@math.utexas.edu}

\author[1]{\fnm{Robert} \sur{Calderbank}}\email{robert.calderbank@duke.edu}

\affil[1]{\orgdiv{Department of Electrical \& Computer Engineering}, \orgname{Duke University}, \orgaddress{\city{Durham}, \state{NC}, \country{USA}}} 

\affil[2]{\orgdiv{Department of Electrical Engineering}, \orgname{Indian Institute of Technology}, \orgaddress{\city{Delhi}, \country{India}}} 

\affil[3]{\orgdiv{Department of Mathematics}, \orgname{University of Texas}, \orgaddress{\city{Austin}, \state{TX}, \country{USA}}} 

\affil[4]{{\orgname{Cohere Technologies Inc.}, \orgaddress{\city{Santa Clara}, \state{CA}, \country{USA}}}} 


\abstract{Zak-OTFS is modulation scheme where signals are formed in the delay-Doppler (DD) domain, converted to the time domain (DD) for transmission and reception, then returned to the DD domain for processing. We describe how to use the same architecture for radar sensing. The intended delay resolution is $\nicefrac{1}{B}$ where $B$ is the radar bandwidth, and the intended Doppler resolution is $\nicefrac{1}{T}$ where $T$ is the transmission time. We form a radar waveform in the DD domain, illuminate the scattering environment, match filter the return, then correlate with delay and Doppler shifts of the transmitted waveform. This produces an image of the scattering environment, and the radar ambiguity function expresses the blurriness of this image. The possible delay and Doppler shifts generate the continuous Heisenberg-Weyl group which has been widely studied in the theory of radar. We describe how to approach the problem of waveform design, not from the perspective of this continuous group, but from the perspective of a discrete group of delay and Doppler shifts, where the discretization is determined by the intended delay and Doppler resolution of the radar. We describe how to approach the problem of shaping the ambiguity surface through symplectic transformations that normalize our discrete Heisenberg-Weyl group. The complexity of traditional continuous radar signal processing is $\mathcal{O}\big(B^2T^2\big)$. We describe how to reduce this complexity to $\mathcal{O}\big(BT\log T\big)$ by choosing the radar waveform to be a common eigenvector of a maximal commutative subgroup of our discrete Heisenberg-Weyl group. The theory of symplectic transformations also enables defining libraries of optimal radar waveforms with small peak-to-average power ratios.}




\maketitle

\section{Introduction}
\label{sec:intro}

The radar ambiguity function expresses blurriness of a scattering environment as a result of illumination by a radar waveform and processing of the return by matched filtering with the transmitted waveform. We approach the problem of designing a radar waveform through the framework of the discrete Heisenberg-Weyl group. The continuous Heisenberg-Weyl groups have a long history in physics~\cite{Folland2016}, in the theory of radar~\cite{Auslander1985,Miller1991,Moran2001_mathofradar,Moran2004_grouptheory_radar}, and in signal processing~\cite{dzt}. This theory features many mathematically sophisticated connections to harmonic analysis. Their discrete time counterparts~\cite{Segal1963_finiteHW,Weil1964_finiteHW,Mackey1965_finiteHW} are more accessible but have received less attention, with the exception of~\cite{Tolimieri1997_finiteHW,Richman1998_finiteHW}.

We start from the Zak-OTFS (orthogonal frequency space modulation) carrier waveform\footnote{Although this paper develops the theory of discrete radar with Zak-OTFS, the results are realizable on practical hardware, e.g., see~\cite{AFRL2025} for an example of a hardware demonstration of sensing using Zak-OTFS.}~\cite{otfs_book,bitspaper1,bitspaper2} and use the Heisenberg-Weyl group to construct a large family of radar waveforms. The Zak-OTFS carrier waveform is a pulse in the delay-Doppler (DD) domain, more specifically a quasi-periodic localized function defined by a delay period $\tau_p$ and a Doppler period $\nu_p$ where $\tau_p \nu_p = 1$. We then apply a DD domain pulse shaping filter. The TD representation of the filtered signal is a TD pulsone, that is a train of narrow impulses modulated by a sinusoid, with time duration $T$ and bandwidth $B$ inversely proportional to the Doppler and delay spread respectively of the filter. In Zak-OTFS modulation, we transmit information using non-overlapping DD domain pulses spaced $\nicefrac{1}{B}$ apart along the delay axis and $\nicefrac{1}{T}$ apart along the Doppler axis. Since each pulse repeats quasi-periodically, there are $M = \nicefrac{\tau_p}{\nicefrac{1}{B}} = B\tau_p$ pulse locations along the delay axis and $N = \nicefrac{\nu_p}{\nicefrac{1}{T}} = T\nu_p$ pulse locations along the Doppler axis. The number of distinct non-overlapping information carriers is the time-bandwidth product $BT = MN$, and arithmetic modulo $MN$ plays an essential role in the mathematical development of our framework for discrete radar. Throughout this paper we assume, for simplicity, that $M$, $N$ are distinct odd primes. We design radar waveforms by applying discrete affine Fourier transforms to a Zak-OTFS carrier.

The inverse Zak transform converts a DD domain waveform to a time domain (TD) waveform. The discrete Zak transform (DZT)~\cite{dzt} defines a unitary equivalence between the Hilbert space of $MN$-periodic sequences in the TD and the Hilbert space of quasi-periodic $M \times N$ arrays in the DD domain. Note that a quasi-periodic $M \times N$ array is periodic modulo $MN$. Section~\ref{subsec:hw_group_def} describes how the Heisenberg-Weyl group acts on both Hilbert spaces, and how the DZT respects both group actions. 

It is customary in radar theory to regard a waveform as a signal that is modulated onto a carrier and to separate carrier modulation and demodulation from the analysis of ambiguity. Separating the choice of the modulating signal from the choice of carrier makes the problem of waveform design modular, and this modularity is at the heart of conventional processing for pulse Doppler radars. In this paper, we focus not on sequence design, but on designing carrier waveforms with target ambiguity properties. 

In Section~\ref{subsec:impl_radar_wvf}, we compare our direct approach to radar waveform design with sequence design followed by modulation onto a carrier. We compare against constant amplitude (CA) sequences with zero autocorrelation (ZAC), known as CAZAC sequences, which are widely used to design discrete time radar waveforms (see~\cite{benedetto_phasecoded}). The family of CAZAC sequences includes Zadoff-Chu (ZC) sequences~\cite{zadoff_chu} that are used in 5G wireless standards because of their correlation properties. By applying a discrete affine Fourier transform to a Zak-OTFS carrier we are able to construct a discrete chirp waveform that coincides with the ZC sequence when sampled appropriately. When we compare the magnitude of the self-ambiguity function of our discrete chirp waveform with the magnitude of the self-ambiguity function of the ZC phase coded rectangular waveform, we see a $100$ dB difference in the level of the sidelobes.

The phase of the Zak-OTFS carrier waveform changes when the pulse location shifts by an integer multiple of $\tau_p$ along the delay axis, and there is no change in phase when the pulse location shifts by an integer multiple of $\nu_p$ along the Doppler axis. We study the ambiguity properties of a general radar waveform $\mathbf{v}$ through its isotropy group $\mathcal{I}_{\mathbf{v}}$ which is the subgroup of the discrete Heisenberg-Weyl group that acts on $\mathbf{v}$ by multiplication by a complex phase. The isotropy group of the Zak-OTFS carrier has size $MN$, generated by a delay shift of $M\tau_p$ and a Doppler shift of $N\nu_p$. The Zak-OTFS carriers form an orthonormal basis, the discrete Heisenberg-Weyl group permutes the carriers, the self-ambiguity function takes the value $1$ when the delay shift is a multiple of $M$ and the Doppler shift is a multiple of $N$, and it takes the value $0$ elsewhere. In general, the self-ambiguity function of a radar waveform $\mathbf{v}$ takes the value $1$ on the isotropy group $\mathcal{I}_{\mathbf{v}}$ and is constant on cosets of $\mathcal{I}_{\mathbf{v}}$ in the discrete Heisenberg-Weyl group. We understand the structure of the ambiguity function of a radar waveform $\mathbf{v}$ by understanding the orbits of $\mathbf{v}$ under the discrete Heisenberg-Weyl group, and we understand the orbits by understanding the isotropy group $\mathcal{I}_{\mathbf{v}}$.

Isotropy groups must be commutative, and a commutative subgroup of the discrete Heisenberg-Weyl group has size at most $MN$. The operators in a maximal commutative subgroup of size $MN$ can be simultaneously diagonalized, giving an orthonormal basis of common eigenvectors. The isotropy subgroup of any eigenvector has size $MN$, and the self-ambiguity function takes the value $1$ on the isotropy group and takes the value $0$ elsewhere. By selecting radar waveforms that are eigenvectors of maximal commutative subgroups we minimize the support of the self-ambiguity function, and we simplify target detection by reducing ambiguity. When we select the radar waveform to be a Zak-OTFS carrier the self-ambiguity function is supported on points $(aM, bN)$ where $a$, $b$ are arbitrary integers (the self-ambiguity function is periodic modulo $MN$ and in Section~\ref{subsec:comm_subgroups} we show that this support set has the structure of a line modulo $MN$). The intended delay resolution $\nicefrac{1}{B}$ and intended Doppler resolution $\nicefrac{1}{T}$ of the radar determine the delay and Doppler periods that define the Zak-OTFS modulation. We have shown in~\cite{bitspaper1,bitspaper2} that if we select the radar waveform to be a Zak-OTFS carrier, then we can separate targets if their delays differ by more than $\nicefrac{1}{B}$ or if their Doppler shifts differ by more than $\nicefrac{1}{T}$. In fact, we can read off an image of the scattering environment\footnote{This paper is limited to forming radar images only in the presence of noise. Extensions to radar image formation in the presence of clutter~\cite{Skolnik1980,Spafford1968_clutter,Shnidman1999_clutter,Shnidman2005_clutter,Calderbank2007_seaclutter} are possible, but are not pursued in this paper.} from the response to a single Zak-OTFS carrier waveform~\cite{bitspaper1,bitspaper2,zakotfs_ltv}.

In Section~\ref{subsec:gdaft} we introduce symplectic operators that normalize the discrete Heisenberg-Weyl group. Conjugation by a symplectic operator defines an isomorphism of the Heisenberg-Weyl group that maps maximal commutative subgroups to maximal commutative subgroups. If a symplectic operator $\mathcal{W}$ maps an eigenvector $\mathbf{v}$ to $\mathcal{W}\mathbf{v}$ then the isotropy group of $\mathcal{W}\mathbf{v}$ is simply $\mathcal{W}\mathcal{I}_{\mathbf{v}}\mathcal{W}^{-1}$. The self-ambiguity function of $\mathbf{v}$, denoted $\mathbf{A}_{\mathbf{v}}[k,l]$ is a function of a delay shift $k$ and a Doppler shift $l$, and is periodic modulo $MN$. In Section~\ref{subsec:gdaft} we introduce discrete affine Fourier transforms, prove that they normalize the discrete Heisenberg-Weyl group, and describe how they transform ambiguity functions. We show that a discrete affine Fourier transform $\mathcal{W}$ determines a $2 \times 2$ matrix $g$ with entries in $\mathbb{Z}_{MN}$ and determinant $1$ such that $\big|\mathbf{A}_{\mathcal{W}\mathbf{v}}[k,l]\big| = \big|\mathbf{A}_{\mathbf{v}}[g \cdot (k,l)]\big|$. Symplectic operators rotate ambiguity functions, and the development in Section~\ref{subsec:gdaft} is the discrete counterpart of the continuous theory used to construct waveform libraries~\cite{Moran2004_wvf_lib,Moran2009_wvf_lib_survey,Evans2002,Moran2006}. These operators also map the optimal radar waveforms to counterparts with small peak-to-average power ratio (PAPR).

Section~\ref{sec:impl_radar} describes how to reduce the complexity of radar signal processing by taking advantage of underlying discrete structure in the radar waveform. Our baseline is traditional continuous radar processing (Figs.~\ref{fig:block_diag}(\subref{fig:cont_radar}) \&~\ref{fig:block_diag}(\subref{fig:disc_phase_coded_radar})) where the complexity of computing the radar cross-ambiguity function is $\mathcal{O}\big(B^2T^2\big)$ dominates signal processing complexity. We propose a discrete architecture that is based on choosing a radar waveform that is highly symmetric –- the isotropy subgroup of the waveform is a maximal commutative subgroup of the Heisenberg-Weyl group. The received waveform is matched filtered to optimize SNR, then sampled and periodized in the time-domain, resulting in a periodic discrete-time signal with period $BT = MN$. The DD domain locations of the targets are obtained by calculating the discrete cross-ambiguity between the discrete-time $MN$-periodic transmitted waveform and the sampled and periodized received signal. Utilizing a Zak-OTFS carrier (or its symplectic transformations) as a radar waveform enables low-complexity cross-ambiguity calculations at a \emph{sublinear complexity} of $\mathcal{O}\big(BT\log T\big)$. Our discrete architecture makes it possible to improve delay and Doppler resolution for a fixed complexity budget. 

There is an additional benefit to working with discrete chirps modulo $MN$ within our framework. The self-ambiguity function of a continuous chirp is a line, whereas the self-ambiguity functions of our chirps are supported on non-contiguous sets in the DD domain. This enables simultaneous resolution in both delay and Doppler, which is not possible with continuous chirps. For more information about target resolution using chirp waveforms and Zak-OTFS waveforms, we refer the reader to~\cite{zakotfs_ltv}.

\textit{Notation:} $x$ denotes a complex scalar, $\mathbf{x}$ denotes a vector with $n$th entry $\mathbf{x}[n]$, and $\mathbf{X}$ denotes a matrix with $(n,m)$th entry $\mathbf{X}[n,m]$. $(\cdot)^{\ast}$ denotes complex conjugate, $(\cdot)^{\top}$ denotes transpose, $(\cdot)^{\mathsf{H}}$ denotes complex conjugate transpose and $\langle \mathbf{x}, \mathbf{y} \rangle = \sum_{n} \mathbf{x}[n] \mathbf{y}^{\ast}[n]$ denotes the inner product. Calligraphic font $\mathcal{X}$ denotes operators or sets, with usage clear from context. $\emptyset$ denotes the empty set. $\mathbb{Z}$ denotes the set of integers and $\mathbb{Z}_{N}$ the set of integers modulo $N$. $(a,b)$ denotes the greatest common divisor of two integers $a,b$, $\lfloor \cdot \rfloor$ and $\lceil \cdot \rceil$ denote the floor and ceiling functions. $(\cdot)_{{}_{N}}$ denotes the value modulo $N$ and $(\cdot)^{-1}_{{}_{N}}$ denotes the inverse modulo $N$. $\delta(\cdot)$ denotes the delta function, $\delta[\cdot]$ denotes the Kronecker delta function, $\mathds{1}{\{\cdot\}}$ denotes the indicator function, and $\mathbf{e}_{n}$ denotes the standard basis vector with value $1$ at location $n$ and zero elsewhere.

\section{Background}
\label{sec:prelim}


This Section prepares the ground for our development of discrete radar based on Zak-OTFS. We use the following identity from classical number theory repeatedly.



\begin{identity}[\cite{iwaniec2021_numtheorybook}, pg.18]
    \label{idty:sumrootsofunity}
    The sum of all $N$th roots of unity satisfies:
    \begin{align*}
        \sum_{n=0}^{N-1}e^{\frac{j2\pi}{N}kn} = \begin{cases}
        N \quad \text{if } \ k \equiv 0 \bmod{N} \\
        0 \quad \ \text{otherwise}
        \end{cases}.
    \end{align*}
\end{identity}


\subsection{Zak-OTFS Signaling}
\label{subsec:prelim_zakotfs}

Here we briefly review how Zak-OTFS signals are formed in the delay-Doppler (DD) domain and are transmitted in the time domain (TD). We refer the interested reader to~\cite{otfs_book,bitspaper1,bitspaper2} for a more detailed description of Zak-OTFS.

The Zak-OTFS carrier waveform is a pulse in the delay-Doppler (DD) domain, formally a quasi-periodic localized function termed the \emph{DD pulsone} defined by a delay period $\tau_p$ and a Doppler period $\nu_p$ where $\tau_p \nu_p = 1$. The DD pulsone occupies infinite time and bandwidth. To enable practical implementation, the DD pulsone is limited to a time interval $T$ and a bandwidth $B$ via DD domain pulse shaping (filtering). The resulting TD representation of the filtered signal is a train of narrow impulses modulated by a sinusoid, with time duration $T$ and bandwidth $B$ inversely proportional to the Doppler and delay spreads of the filter respectively.

Information symbols in Zak-OTFS are transmitted using non-overlapping DD domain pulsones spaced $\nicefrac{1}{B}$ apart along the delay axis and $\nicefrac{1}{T}$ apart along the Doppler axis. Since each pulsone repeats quasi-periodically, this corresponds to mounting information on $M = \nicefrac{\tau_p}{\nicefrac{1}{B}} = B\tau_p$ distinct locations along delay and $N = \nicefrac{\nu_p}{\nicefrac{1}{T}} = T\nu_p$ distinct locations along Doppler. Thus, the number of distinct non-overlapping information carriers in Zak-OTFS is the time-bandwidth product $BT = MN$. The $MN$ information symbols are mounted on pulsones in the DD domain as:
\begin{align}
    \label{eq:sys_model1}
    x(\tau, \nu) = \sum_{k=0}^{M-1} \sum_{l=0}^{N-1} \mathbf{X}[k, l]\sum_{c,d\in\mathbb{Z}} e^{\frac{j2\pi}{N}dl}\delta\left(\tau - \frac{k\tau_p}{M}-d\tau_p\right)\delta\left(\nu -\frac{l\nu_p}{N}-c\nu_p\right),
\end{align}
where $\mathbf{X}\in \mathbb{C}^{M\times N}$ denotes the $M \times N$ DD array of information symbols, $\tau$ and $\nu$ are delay and Doppler axes respectively, and $\delta(\cdot)$ is the delta function. The continuous TD representation of the DD signal in~\eqref{eq:sys_model1} is via the inverse Zak transform:
\begin{align}
    \label{eq:sys_model2}
    x(t) = \mathcal{Z}^{-1}\{x(\tau,\nu)\}(t) = \sqrt{\tau_p}\int_0^{\nu_p}x(t, \nu) d\nu.
\end{align}

Substituting~\eqref{eq:sys_model1} in~\eqref{eq:sys_model2}:
\begin{align}
    \label{eq:sys_model3}
    x(t) = \sqrt{\tau_p}\sum_{k=0}^{M-1}\sum_{l=0}^{N-1}\mathbf{X}[k, l]\sum_{d\in\mathbb{Z}}e^{\frac{j2\pi}{N}dl}\delta\left(t-\frac{k\tau_p}{M}-d\tau_p\right).
\end{align}

Note that~\eqref{eq:sys_model3} is defined in terms of the DD array of information symbols $\mathbf{X}$. Formally, $\mathbf{X}$ is a \emph{quasi-periodic} array\footnote{An $M \times N$ array $\mathbf{X}$ is quasi-periodic if $\mathbf{X}[k+nM,l+mN] = e^{j2\pi\frac{nl}{N}} \mathbf{X}[k,l], \text{for all~} n,m \in \mathbb{Z}$.} (see~\cite{dzt,otfs_book,bitspaper1,bitspaper2} for more details). We now introduce the discrete Zak transform (DZT), which is a unitary map from $MN$-periodic discrete time sequences $\mathbf{x}$ to $M \times N$ quasi-periodic arrays $\mathbf{X}$ (see~\cite{dzt} for details). This fundamental equivalence makes it possible to implement signal processing functions in either the time domain or the delay-Doppler domain. 

\begin{definition}[\cite{dzt}]
    \label{def:dzt}
    The discrete Zak transform (DZT) maps an $MN$-periodic sequence $\mathbf{x}$ to a quasi-periodic $M \times N$ array $\mathbf{X}$ as follows:
    \begin{align*}
        \mathbf{X}[k, l] = \frac{1}{\sqrt{N}} \sum_{p = 0}^{N-1}\mathbf{x}[k+pM]e^{-\frac{j2\pi}{N} pl},
    \end{align*}
    where $k \in \mathbb{Z}_{M}$ is the delay and $l \in \mathbb{Z}_{N}$ is the Doppler index. The $M \times N$ quasi-periodic arrays $\mathbf{X}[k, l]$ form a Hilbert space, and inner products can be calculated on any $M \times N$ subarray. Similarly, the $MN$-periodic sequences $\mathbf{x}$ form a Hilbert space, and inner products can be calculated on any interval of length $MN$.
\end{definition}


Substituting Definition~\ref{def:dzt} in~\eqref{eq:sys_model3}:
\begin{align}
    \label{eq:sys_model5}
    x(t) &= \sqrt{\tau_p}\sum_{k=0}^{M-1}\sum_{l=0}^{N-1}\mathbf{X}[k, l]\sum_{d\in\mathbb{Z}}e^{\frac{j2\pi}{N}dl}\delta\left(t-\frac{k\tau_p}{M}-d\tau_p\right) \nonumber \\
    &= \sqrt{\frac{\tau_p}{N}} \sum_{k=0}^{M-1}\sum_{p = 0}^{N-1}\sum_{d\in\mathbb{Z}} \mathbf{x}[k+pM] \delta\left(t-\frac{k\tau_p}{M}-d\tau_p\right) \sum_{l=0}^{N-1} e^{\frac{j2\pi}{N}(d-p)l}.
\end{align}

From Identity~\ref{idty:sumrootsofunity}, the inner summation over $l$ vanishes unless $d \equiv p \bmod{N}$, when it takes the value $N$. Hence, we obtain:
\begin{align}
    \label{eq:sys_model6}
    x(t) &= \sqrt{N\tau_p} \sum_{k=0}^{M-1}\sum_{p = 0}^{N-1} \mathbf{x}[k+pM] \delta\left(t-\frac{k\tau_p}{M}-p\tau_p\right) \nonumber \\
    &= \sqrt{N\tau_p} \sum_{n=0}^{MN-1} \mathbf{x}[n] \delta\left(t-\frac{(n)_{{}_{M}}\tau_p}{M}-\bigg\lfloor \frac{n}{M} \bigg\rfloor \tau_p\right) \nonumber \\
    &= \sqrt{N\tau_p} \sum_{n=0}^{MN-1} \mathbf{x}[n] \delta\left(t-\frac{n\tau_p}{M}\right).
\end{align}

We can view the continuous TD signal $x(t)$ as a discrete time sequence $\mathbf{x}$ mounted on delta functions spaced at the delay resolution $\nicefrac{1}{B} = \nicefrac{\tau_p}{M}$. This fact is central to our proposed framework for discrete radar based on Zak-OTFS. 

\subsection{Discrete Radar System Model}
\label{subsec:prelim_radar}

At a high-level, the operation of a radar consists of transmitting a waveform, which is reflected by the scattering environment and its return is processed at the receiver~\cite{Woodward1953,Skolnik1980,Levanon2004}. Representing the scattering environment by a DD function $h(\tau,\nu)$ and the continuous TD transmit signal by $x(t)$, the continuous TD received signal is:
\begin{align}
    \label{eq:radar1}
    y(t) &= \iint h(\tau,\nu) x(t-\tau) e^{j2\pi\nu(t-\tau)} d\tau d\nu,
\end{align}
where we have neglected additive noise at the receiver. 

To derive the system model corresponding to a discrete radar, we map the transmit and received continuous TD signals to $MN$-periodic sequences by projecting the continuous TD signals onto the $\nicefrac{1}{B} = \nicefrac{\tau_p}{M}$-spaced delta function basis from~\eqref{eq:sys_model6}. Specifically, from~\eqref{eq:radar1} we compute:
\begin{align}
    \label{eq:radar2}
    \mathbf{y}[n] &= \int y(t) \delta\left(t-\frac{n\tau_p}{M}\right) dt \nonumber \\
    &= \iint h(\tau,\nu) \underbrace{\bigg(\int x(t-\tau) e^{j2\pi\nu(t-\tau)} \delta\left(t-\frac{n\tau_p}{M}\right) dt\bigg)}_{\mathsf{A}} d\tau d\nu.
\end{align}

To evaluate the integral $\mathsf{A}$ we substitute~\eqref{eq:sys_model6}:
\begin{align}
    \label{eq:radar3}
    \mathsf{A} &= \sqrt{N\tau_p} \sum_{m=0}^{MN-1} \mathbf{x}[m] \int e^{j2\pi\nu(t-\tau)} \delta\left(t-\tau-\frac{m\tau_p}{M}\right) \delta\left(t-\frac{n\tau_p}{M}\right) dt.
\end{align}

To evaluate the integral in~\eqref{eq:radar3}, we assume the delay $\tau$ and Doppler $\nu$ are integer multiples of the delay and Doppler resolutions, i.e., $\tau = \nicefrac{k\tau_p}{M}$ for some $k \in \mathbb{Z}_{MN}$, and $\nu = \nicefrac{l\nu_p}{N}$ for some $l \in \mathbb{Z}_{MN}$. Substituting into~\eqref{eq:radar3}: 
\begin{align}
    \label{eq:radar4}
    \mathsf{A} &= \sqrt{N\tau_p} \sum_{m=0}^{MN-1} \mathbf{x}[m] \int e^{\frac{j2\pi}{N}l\big(t\nu_p-\frac{k}{M}\big)} \delta\left(t-\frac{(m+k)\tau_p}{M}\right) \delta\left(t-\frac{n\tau_p}{M}\right) dt \nonumber \\
    &= \sqrt{N\tau_p} \mathbf{x}[(n-k)_{{}_{MN}}] e^{\frac{j2\pi}{MN}l(n-k)},
\end{align}
where the final expression follows from the conditions $(m+k) \equiv n \bmod{MN}$ and $t = \nicefrac{n\tau_p}{M}$ obtained from the delta functions. Substituting the value for $\mathsf{A}$ back into~\eqref{eq:radar2}:
\begin{align}
    \label{eq:radar5}
    \mathbf{y}[n] &= \sqrt{N\tau_p} \sum_{k,l \in \mathbb{Z}_{MN}} h\bigg(\frac{k\tau_p}{M},\frac{l\nu_p}{N}\bigg) \mathbf{x}[(n-k)_{{}_{MN}}] e^{\frac{j2\pi}{MN}l(n-k)} \nonumber \\
    &= \sum_{k,l \in \mathbb{Z}_{MN}} \mathbf{h}[k,l] \mathbf{x}[(n-k)_{{}_{MN}}] e^{\frac{j2\pi}{MN}l(n-k)},
\end{align}
where $\mathbf{h}[k,l] = \sqrt{N\tau_p} h\big(\nicefrac{k\tau_p}{M},\nicefrac{l\nu_p}{N}\big)$ is the scaled \& sampled version of the DD function $h(\tau,\nu)$ at multiples of the delay and Doppler resolutions. Thus, as per~\eqref{eq:radar5}, the output of a discrete radar is an $MN$-periodic sequence $\mathbf{y}$ given by the weighted sum of various delay-shifted and Doppler-modulated versions of the transmit sequence $\mathbf{x}$.

\textcolor{black}{\textit{Note:} Although our derivation above utilizes delta functions in the transmit waveform, in practice, the function $h(\tau,\nu)$ takes into account both pulse shaping and the physical scattering environment, which is referred to as the \emph{effective channel} in the literature~\cite{otfs_book,bitspaper1,bitspaper2}. Pulse shaping helps limit the time and bandwidth of the waveform to $T$ and $B$ respectively. Moreover, our framework is capable of considering paths in the physical scattering environment with \emph{fractional} delay and Doppler values as long as the effective channel is sampled at multiples of the delay and Doppler resolutions. For more details, we refer the interested reader to~\cite{otfs_book,bitspaper1,bitspaper2,Calderbank2025_isac,zakotfs_ltv,Mehrotra2025_WCLSpread}.}

\subsection{Ambiguity Functions}
\label{subsec:prelim_ambgfun}


Given the system model in~\eqref{eq:radar5}, radars generate a discrete image of the scattering environment by cross-correlating the received sequence $\mathbf{y}$ with different delay-shifted and Doppler-modulated versions of the transmit sequence $\mathbf{x}$. Formally, this is accomplished via the \emph{cross-ambiguity} function, defined below for sequences in the Hilbert space $\mathcal{H}$ of all complex-valued $MN$-periodic sequences. Note that since the sequences are $MN$-periodic, inner products can be calculated within any window of length $MN$.

\begin{definition}[\cite{benedetto_phasecoded}]
    \label{def:amb_fun}
    The \emph{cross-ambiguity function} of two unit-norm $MN$-periodic sequences $\mathbf{x},\mathbf{y} \in \mathcal{H}$ is defined as:
    \begin{align*}
        \mathbf{A}_{\mathbf{x},\mathbf{y}}[k,l] = \sum_{n=0}^{MN-1} \mathbf{x}[n] \mathbf{y}^{*}[(n-k)_{{}_{MN}}]e^{-\frac{j2\pi}{MN}l(n-k)}, 
    \end{align*}
    where $k,l \in \mathbb{Z}_{MN}$, the integers modulo $MN$. When $\mathbf{y} = \mathbf{x}$, $\mathbf{A}_{\mathbf{x},\mathbf{x}}[k, l]$ is referred to as the \emph{self-ambiguity function} of $\mathbf{x}$, and is abbreviated to $\mathbf{A}_{\mathbf{x}}[k, l]$ for brevity.
\end{definition}

Radars represent the scattering environment by calculating from~\eqref{eq:radar5}:
\begin{align}
    \label{eq:radar6}
    \widehat{\mathbf{h}}[k,l] &= \mathbf{A}_{\mathbf{y},\mathbf{x}}[k,l] \nonumber \\
    &= \sum_{n=0}^{MN-1} \mathbf{y}[n] \mathbf{x}^{*}[(n-k)_{{}_{MN}}]e^{-\frac{j2\pi}{MN}l(n-k)} \nonumber \\
    &= \sum_{k',l' \in \mathbb{Z}_{MN}} \mathbf{h}[k',l'] \sum_{n=0}^{MN-1} \mathbf{x}[(n-k')_{{}_{MN}}] \mathbf{x}^{*}[(n-k)_{{}_{MN}}] e^{\frac{j2\pi}{MN}\big[l'(n-k')-l(n-k)\big]} \nonumber \\
    &= \sum_{k',l' \in \mathbb{Z}_{MN}} \mathbf{h}[k',l'] \sum_{n'} \mathbf{x}[n'] \mathbf{x}^{*}[(n'-(k-k'))_{{}_{MN}}] e^{\frac{j2\pi}{MN}\big[(l'-l)n'+l(k-k')\big]} \nonumber \\
    &= \sum_{k',l' \in \mathbb{Z}_{MN}} \mathbf{h}[k',l'] e^{\frac{j2\pi}{MN}l'(k-k')} \mathbf{A}_{\mathbf{x}}[(k-k')_{{}_{MN}}, (l-l')_{{}_{MN}}].
\end{align}

The above expression shows that the image formed by the radar is essentially the DD function $\mathbf{h}[k,l]$ ``blurred'' by the (phase-scaled) self-ambiguity function of the transmitted sequence $\mathbf{x}$. Ideally, it would be desirable to choose a sequence $\mathbf{x}$ with a ``thumbtack-like'' self-ambiguity function. However, fundamental properties of the ambiguity function limit the amount of waveform shaping that is possible. Specifically, Moyal's Identity stated below shows that though the shape of the self-ambiguity function surface depends on the choice of waveform, the volume under the surface always remains fixed and equals the (scaled) energy of the waveform. 

\begin{identity}[Moyal's Identity~\cite{Moyal1949,Auslander1985,Miller1991,Moran2001_mathofradar,Moran2004_grouptheory_radar,Howard2006_HW}]
    \label{idty:moyal}
    The self-ambiguity functions of unit-norm $MN$-periodic sequences $\mathbf{x},\mathbf{y} \in \mathcal{H}$ satisfy:
    \begin{align*}
        \frac{1}{MN} \sum_{k = 0}^{MN-1} \sum_{l = 0}^{MN-1} \mathbf{A}_{\mathbf{x}}^{*}[k, l] \mathbf{A}_{\mathbf{y}}[k, l] = \big|\big\langle \mathbf{x}, \mathbf{y} \big\rangle\big|^{2}.
    \end{align*}

    A special case of the above relation corresponds to $\mathbf{y} = \mathbf{x}$:
    \begin{align*}
        \frac{1}{MN} \sum_{k = 0}^{MN-1} \sum_{l = 0}^{MN-1} \big|\mathbf{A}_{\mathbf{x}}[k, l]\big|^{2} = 1.
    \end{align*}
\end{identity}

A direct corollary of Moyal's Identity is the following.

\begin{corollary}
    \label{corr:moyal_setcard}
    Given a unit-norm $MN$-periodic sequence $\mathbf{x} \in \mathcal{H}$, there are at most $MN$ delay-Doppler indices $(k,l)$, $k \in \mathbb{Z}_{MN}$, $l \in \mathbb{Z}_{MN}$, for which $\big|\mathbf{A}_{\mathbf{x}}[k, l]\big| = 1$.
\end{corollary}
\begin{proof}
    Let $\mathcal{S} = \big\{(k,l)\big| \big|\mathbf{A}_{\mathbf{x}}[k, l]\big| = 1 \big\}$ denote the set of delay-Doppler indices where the ambiguity function is unimodular. It follows from Identity~\ref{idty:moyal} that:
    \begin{align*}
        \frac{1}{MN} \sum_{(k,l) \in \mathcal{S}} \big|\mathbf{A}_{\mathbf{x}}[k, l]\big|^{2} = \frac{\mathsf{card}\big(\mathcal{S}\big)}{MN} \leq \frac{1}{MN} \sum_{k = 0}^{MN-1} \sum_{l = 0}^{MN-1} \big|\mathbf{A}_{\mathbf{x}}[k, l]\big|^{2} = 1.
    \end{align*}
\end{proof}

Later in Sections~\ref{subsec:eigenbasis} and~\ref{subsec:impl_radar_wvf} we demonstrate how the proposed framework for discrete radar based on Zak-OTFS enables achieving perfectly localized ambiguity functions within a subset of the DD grid $(k,l) \in \mathbb{Z}_{MN} \times \mathbb{Z}_{MN}$.

\section{A Group Theoretic Framework for Discrete Radar}
\label{sec:hw_group}

In this Section, we show how the basic discrete radar theory described in Sections~\ref{subsec:prelim_radar} and~\ref{subsec:prelim_ambgfun} can be studied under a more general group-theoretic framework. We show how the action of the scattering environment in~\eqref{eq:radar5} corresponds to the superposed actions of the discrete Heisenberg-Weyl group. We then derive various useful properties of the discrete Heisenberg-Weyl group. We explain how the theory applies to two prototypical examples, which we follow through the Section. In Section~\ref{sec:impl_radar} we explore how the theory developed in this Section enables: (i) designing radar waveforms with ideal self-ambiguity functions, (ii) computing images of the scattering environment at low complexity, and (iii) defining radar waveform libraries that can be implemented at low hardware complexity due to their small peak-to-average power ratios. 

For simplicity, we henceforth assume that the integers $M$ and $N$ introduced in Section~\ref{sec:prelim} are odd primes\footnote{\textcolor{black}{In this paper, we focus on the value of choosing a radar waveform to be an eigenvector of a maximal commutative subgroup (see Section~\ref{subsec:eigenbasis} for more details). When $M$, $N$ are allowed to be composite and to have common factors, more types of maximal commutative subgroups become possible, and the support of the self-ambiguity function becomes more varied. In the analysis, sums of $MN$th roots of unity need to be replaced by sums of $(\nicefrac{MN}{d})$th roots of unity, where $d = (M,N)$ is the greatest common divisor of $M$ and $N$. In summary, the principle does not change, but the details become more complicated.}} (unless mentioned otherwise). 

\subsection{The Discrete Heisenberg-Weyl Group}
\label{subsec:hw_group_def}

We begin by recalling the action of the scattering environment on the transmitted sequence in~\eqref{eq:radar5}, which is reproduced here for ease of reference:
\begin{align}
    \label{eq:hw_op}
    \mathbf{y}[n] &= \sum_{k,l \in \mathbb{Z}_{MN}} \mathbf{h}[k,l] \mathbf{x}[(n-k)_{{}_{MN}}] e^{\frac{j2\pi}{MN}l(n-k)}.
\end{align}

For a given $k,l \in \mathbb{Z}_{MN}$, we may represent the action of the scattering environment by a unitary operator that delay shifts and Doppler modulates any input sequence belonging to the Hilbert space $\mathcal{H}$ of $MN$-periodic sequences.



\begin{definition}
    \label{def:heis_op}
    \emph{Heisenberg operators} $\mathcal{D}_{(k,l)}$ act on sequences ${\mathbf{x} \in \mathcal{H}}$ as:
    \begin{equation*}
        \begin{gathered}
        \mathcal{D}_{(k,l)}: \mathcal{H} \rightarrow \mathcal{H}, k,l \in \mathbb{Z}_{MN}, \\
        (\mathcal{D}_{(k,l)} \mathbf{x})[n] = \mathbf{x}[(n-k)_{{}_{MN}}]e^{\frac{j2\pi}{MN}l(n-k)},
        \end{gathered}
    \end{equation*}
    for all $n \in \mathbb{Z}_{MN}$. 
\end{definition}

Heisenberg operators are unitary operators since: 
\begin{equation*}
    \Vert \mathcal{D}_{(k,l)} \mathbf{x} \Vert_{2}^{2} = \sum_{n \in \mathbb{Z}_{MN}} \big|\mathbf{x}[(n-k)_{{}_{MN}}] \big|^{2} = \Vert \mathbf{x} \Vert_{2}^{2}.
\end{equation*}

The collection of all (phase-scaled) Heisenberg operators forms a group.

\begin{theorem}
    \label{thm:hw_group}
    Let $\mathcal{H}_{MN}$ denote the collection of phase-scaled Heisenberg operators $\mathcal{D}_{(k,l)}$:
    \begin{align*}
        \mathcal{H}_{MN} &= \big\{e^{\frac{j2\pi}{MN}m} \mathcal{D}_{(k,l)} \big| k,l,m \in \mathbb{Z}_{MN}\big\}. 
    \end{align*}

    The set $\mathcal{H}_{MN}$ under operator composition $\circ$ forms the \emph{Heisenberg-Weyl group}, which can be equivalently defined in terms of $3 \times 3$ upper triangular matrices:
    \begin{align*}
        \mathcal{H}_{MN} &= \left\{\left.\begin{bmatrix}
            1 & l & m \\ 0 & 1 & k \\ 0 & 0 & 1
        \end{bmatrix} \right| k,l,m \in \mathbb{Z}_{MN} \right\},
    \end{align*}
    under matrix multiplication modulo $MN$.
\end{theorem}

\begin{proof}
    We verify that the Heisenberg-Weyl group satisfies all group axioms.

    \textbf{(Closure)} For elements $h_1 = e^{\frac{j2\pi}{MN}m_1} \mathcal{D}_{(k_1,l_1)} \in \mathcal{H}_{MN}$, $h_2 = e^{\frac{j2\pi}{MN}m_2} \mathcal{D}_{(k_2,l_2)} \in \mathcal{H}_{MN}$, we compute:
    \begin{align}
        \label{eq:hw2a}
        (h_1 \circ h_2) &= \begin{bmatrix}
            1 & l_1 & m_1 \\ 0 & 1 & k_1 \\ 0 & 0 & 1
        \end{bmatrix} \begin{bmatrix}
            1 & l_2 & m_2 \\ 0 & 1 & k_2 \\ 0 & 0 & 1
        \end{bmatrix} \nonumber \\
        &= \begin{bmatrix}
            1 & (l_1+l_2) & (m_1+m_2+l_1 k_2) \\ 0 & 1 & (k_1+k_2) \\ 0 & 0 & 1
        \end{bmatrix}.
    \end{align}
    
    Since $(h_1 \circ h_2) = e^{\frac{j2\pi}{MN}(m_1+m_2+l_1 k_2)} \mathcal{D}_{((k_1+k_2)_{MN},(l_1+l_2)_{MN})} \in \mathcal{H}_{MN}$, the group is closed.
    
    \textbf{(Associativity)} Since operator composition (equivalently, matrix multiplication) is associative, $h_1 \circ (h_2 \circ h_3) = (h_1 \circ h_2) \circ h_3$.
    
    
    \textbf{(Identity)} The identity element is given by $\mathcal{D}_{(0,0)} \in \mathcal{H}_{MN}$, which corresponds to the identity matrix in the matrix description.
    
    
    

    \textbf{(Inverse)} Equating~\eqref{eq:hw2a} to the identity matrix and solving for $h_1$ for a given $h_2$, we obtain the conditions:
    \begin{align}
        \label{eq:hw2b}
        l_1 + l_2 \equiv 0 \bmod{MN} &\implies l_1 = (-l_2)_{{}_{MN}}, \nonumber \\
        k_1 + k_2 \equiv 0 \bmod{MN} &\implies k_1 = (-k_2)_{{}_{MN}}, \nonumber \\
        m_1 + m_2 -l_2 k_2 \equiv 0 \bmod{MN} &\implies m_1 = (l_2 k_2 - m_2)_{{}_{MN}}.
    \end{align}
    
    
    Thus, $e^{\frac{j2\pi}{MN}(lk-m)} \mathcal{D}_{((-k)_{MN},(-l)_{MN})} \in \mathcal{H}_{MN}$ is the inverse of $e^{\frac{j2\pi}{MN}m} \mathcal{D}_{(k,l)} \in \mathcal{H}_{MN}$.
\end{proof}

The following Lemma determines when two elements of $\mathcal{H}_{MN}$ commute.

\begin{lemma}
    \label{lmm:hw_not_comm}
    The symplectic form $l_1 k_2 \equiv l_2 k_1 \bmod{MN}$ determines when two elements $h_1 = e^{\frac{j2\pi}{MN}m_1} \mathcal{D}_{(k_1,l_1)} \in \mathcal{H}_{MN}$ and $h_2 = e^{\frac{j2\pi}{MN}m_2} \mathcal{D}_{(k_2,l_2)} \in \mathcal{H}_{MN}$ commute. The symplectic form is not identically zero and the Heisenberg-Weyl group is \emph{not commutative} in general. 
\end{lemma}

\begin{proof}
    We may express~\eqref{eq:hw2a} as:
    \begin{align}
        \label{eq:hw3a}
        (h_1 \circ h_2) &= \begin{bmatrix}
            1 & (l_1+l_2) & (m_1+m_2+l_1 k_2) \\ 0 & 1 & (k_1+k_2) \\ 0 & 0 & 1
        \end{bmatrix} \nonumber \\
        &= \begin{bmatrix}
            1 & 0 & l_1 k_2 \\ 0 & 1 & 0 \\ 0 & 0 & 1
        \end{bmatrix} \begin{bmatrix}
            1 & (l_1+l_2) & (m_1+m_2) \\ 0 & 1 & (k_1+k_2) \\ 0 & 0 & 1
        \end{bmatrix} \nonumber \\
        &= e^{\frac{j2\pi}{MN}l_1 k_2} e^{\frac{j2\pi}{MN}(m_1+m_2)} \mathcal{D}_{((k_1+k_2)_{MN},(l_1+l_2)_{MN})},
    \end{align}
    and thus:
    \begin{align}
        \label{eq:hw3b}
        (h_1 \circ h_2) = e^{\frac{j2\pi}{MN}(l_1 k_2 - l_2 k_1)} (h_2 \circ h_1).
    \end{align}

    Hence, $(h_1 \circ h_2) = (h_2 \circ h_1)$ when $l_1 k_2 \equiv l_2 k_1 \bmod{MN}$.
\end{proof}

We now place the Heisenberg-Weyl group defined above in the context of the discrete radar theory described in Sections~\ref{subsec:prelim_radar} and~\ref{subsec:prelim_ambgfun}. Note that the cross-ambiguity function defined in Definition~\ref{def:amb_fun} can be expressed in terms of Heisenberg operators $\mathcal{D}_{(k,l)}$ and Heisenberg-Weyl group elements $h = e^{\frac{j2\pi}{MN}m} \mathcal{D}_{(k,l)} \in \mathcal{H}_{MN}$ as:
\begin{align}
    \label{eq:amb_fun_hw}
    \mathbf{A}_{\mathbf{x},\mathbf{y}}[k,l] = \big\langle \mathbf{x}, \mathcal{D}_{(k,l)} \mathbf{y} \big\rangle = e^{\frac{j2\pi}{MN}m} \big\langle \mathbf{x}, h \mathbf{y} \big\rangle.
\end{align}

In other words, the cross-ambiguity function corresponds (up to a phase) to the inner product of a sequence $\mathbf{x}$ with the action of a Heisenberg-Weyl group element $h \in \mathcal{H}_{MN}$ on the sequence $\mathbf{y}$. Hence, properties of the Heisenberg-Weyl group elements directly impact the behavior of the discrete radar ambiguity function. After a brief detour in Section~\ref{subsubsec:hw_group_dd} where we define an equivalent representation of the Heisenberg-Weyl group in the delay-Doppler domain, we investigate various properties satisfied by Heisenberg-Weyl group elements in Sections~\ref{subsec:comm_subgroups},~\ref{subsec:eigenbasis} and~\ref{subsec:gdaft}.



\subsubsection{The Heisenberg-Weyl Group in the Delay-Doppler Domain}
\label{subsubsec:hw_group_dd}

In Theorem~\ref{thm:hw_group}, we defined the Heisenberg-Weyl group in terms of $MN$-periodic sequences. We now derive an equivalent representation of the Heisenberg-Weyl group in the \emph{delay-Doppler domain} by mapping $MN$-periodic sequences to $M \times N$ quasi-periodic arrays through the DZT (Definition~\ref{def:dzt}). This equivalence makes it possible to implement signal processing functions in either the time or delay-Doppler domains.

\begin{lemma}
    \label{lmm:dzt_unitary}
    The DZT in Definition~\ref{def:dzt} is a unitary transform that preserves inner products.
\end{lemma}

\begin{proof}
Consider the following basis for the Hilbert space $\mathcal{H}$ of $MN$-periodic sequences:
\begin{align}
    \label{eq:timebasis}
    \mathbf{v}_{r,s}[n] = \begin{cases}
        \frac{1}{\sqrt{M}} e^{\frac{j2\pi}{M} sn}, \quad \text{if } rM \leq n < (r+1)M \\ 
        0, \quad \text{otherwise},
    \end{cases}
\end{align}
where $\mathbf{v}_{r,s}[n]$ is repeated periodically with period $MN$, which closely resembles the discrete time version of the pulsone basis element presented in~\eqref{eq:sys_model3}. The basis is orthonormal since:
\begin{align*}
    \big\langle \mathbf{v}_{r,s}, \mathbf{v}_{r',s'} \big\rangle &= \frac{1}{M} \sum_{n = 0}^{MN-1} \mathds{1}_{\{rM \leq n < (r+1)M\}} e^{\frac{j2\pi}{M} sn} \mathds{1}_{\{r'M \leq n < (r'+1)M\}} e^{-\frac{j2\pi}{M} s'n} \nonumber \\
    &= \delta[r-r'] \frac{1}{M} \sum_{n = rM}^{(r+1)M-1} e^{\frac{j2\pi}{M} (s-s')n} \nonumber \\
    &= \delta[r-r']  \delta[s-s'],
\end{align*}
where the last expression follows from Identity~\ref{idty:sumrootsofunity}. Per Definition~\ref{def:dzt}, the DZT of~\eqref{eq:timebasis} is:
\begin{align}
    \label{eq:timebasis_dzt}
    \mathbf{V}_{r,s}[k,l] &= \frac{1}{\sqrt{N}} \sum_{p = 0}^{N-1}\mathbf{v}_{r,s}[k+pM] e^{-\frac{j2\pi}{N} pl} \nonumber \\
    &= \frac{1}{\sqrt{MN}} \sum_{p = 0}^{N-1} e^{\frac{j2\pi}{M}s(k+pM)} e^{-\frac{j2\pi}{N} pl} \mathds{1}_{\{rM \leq k+pM < (r+1)M\}} \nonumber \\
    &= \frac{1}{\sqrt{MN}} \sum_{p = 0}^{N-1} e^{\frac{j2\pi}{M}s(k+pM)} \delta\bigg[p-r+\bigg\lfloor\frac{k}{M}\bigg\rfloor\bigg] e^{-\frac{j2\pi}{N} pl} \nonumber \\
    &= \frac{1}{\sqrt{MN}} e^{\frac{j2\pi}{M}s(k+(r-\lfloor\frac{k}{M}\rfloor)M)} e^{-\frac{j2\pi}{N} (r-\lfloor\frac{k}{M}\rfloor)l} \nonumber \\ 
    &= \frac{1}{\sqrt{MN}} e^{\frac{j2\pi}{M}sk} e^{-\frac{j2\pi}{N} (r-\lfloor\frac{k}{M}\rfloor)l}.
\end{align}

The obtained basis is clearly orthonormal since:
\begin{align*}
   \big\langle \mathbf{V}_{r,s}, \mathbf{V}_{r',s'} \big\rangle &= \frac{1}{MN} \sum_{k = 0}^{M-1} \sum_{l = 0}^{N-1} e^{j\frac{2\pi}{M}(s-s')k} e^{-j\frac{2\pi}{N}(r-r')l} \nonumber \\ 
   &= \delta[r-r'] \delta[s-s'],
\end{align*}
where the final result follows from Identity~\ref{idty:sumrootsofunity}. Hence, the DZT preserves inner products since it maps the orthonormal basis $\mathbf{v}_{r,s}$ to the orthonormal basis $\mathbf{V}_{r,s}$.
\end{proof}

Let $\mathcal{H}_{\mathsf{DD}}$ denote the Hilbert space of all complex-valued $M \times N$ quasi-periodic arrays formed by the DZT of $MN$-periodic sequences in $\mathcal{H}$. We now define Heisenberg operators, similar to Definition~\ref{def:heis_op}, in the delay-Doppler (DD) domain.

\begin{definition}
    \label{def:heis_op_dd}
    DD Heisenberg operators $\mathcal{D}_{(k,l)}^{\mathsf{DD}}$ are unitary operators that act on quasi-periodic arrays $\mathbf{X} \in \mathcal{H}_{\mathsf{DD}}$ as:
    \begin{equation*}
        \begin{gathered}
        \mathcal{D}_{(k,l)}^{\mathsf{DD}}: \mathcal{H}_{\mathsf{DD}} \rightarrow \mathcal{H}_{\mathsf{DD}}, k,l \in \mathbb{Z}_{MN}, \\
        (\mathcal{D}_{(k,l)}^{\mathsf{DD}} \mathbf{X})[k',l'] = \mathbf{X}[(k'-k)_{{}_{M}},(l'-l)_{{}_{N}}] e^{\frac{j2\pi}{N}(l'-l)\big\lfloor \frac{k'-k}{M} \big\rfloor} e^{\frac{j2\pi}{MN}l(k'-k)},
        \end{gathered}
    \end{equation*}
    for all $k' \in \mathbb{Z}_{M}, l' \in \mathbb{Z}_{N}$.
\end{definition}

We now define the Heisenberg-Weyl group in the delay-Doppler domain, and show that the group representation in the delay-Doppler domain is unitarily equivalent to the Heisenberg-Weyl group representation given in Theorem~\ref{thm:hw_group}.

\begin{theorem}
    \label{thm:hw_group_dd}
    Let $\mathcal{H}_{MN}^{\mathsf{DD}}$ denote the collection of phase-scaled DD Heisenberg operators $\mathcal{D}_{(k,l)}^{\mathsf{DD}}$:
    \begin{align*}
        \mathcal{H}_{MN}^{\mathsf{DD}} &= \big\{e^{\frac{j2\pi}{MN}m} \mathcal{D}_{(k,l)}^{\mathsf{DD}} \big| k,l,m \in \mathbb{Z}_{MN} \big\}. 
    \end{align*}

    The set $\mathcal{H}_{MN}^{\mathsf{DD}}$ under operator composition $\circ$ forms the Heisenberg-Weyl group in the delay-Doppler domain, which is unitarily equivalent to the Heisenberg-Weyl group representation $\mathcal{H}_{MN}$ in Theorem~\ref{thm:hw_group} over the Hilbert space of $MN$-periodic sequences.
\end{theorem}

\begin{proof}
    It suffices to show that $\mathcal{H}_{MN}^{\mathsf{DD}}$ is unitarily equivalent to $\mathcal{H}_{MN}$ from Theorem~\ref{thm:hw_group}, in which case $\mathcal{H}_{MN}^{\mathsf{DD}}$ inherits all the properties satisfied by $\mathcal{H}_{MN}$. Letting $\mathcal{Z}$ denote the operator representation of the DZT, we must therefore show:
    \begin{align}
        \label{eq:hw_dd1}
        \mathcal{Z} \circ \mathcal{D}_{(k,l)} \circ \mathcal{Z}^{-1} = \mathcal{D}_{(k,l)}^{\mathsf{DD}}.
    \end{align}

    Consider an $M \times N$ quasi-periodic array $\mathbf{X}$. The inverse DZT per Definition~\ref{def:dzt} is given by:
    \begin{align}
        \label{eq:hw_dd2}
        \big(\mathcal{Z}^{-1}\mathbf{X}\big)[n] &= \frac{1}{\sqrt{N}} \sum_{q = 0}^{N-1}\mathbf{X}[(n)_{{}_{M}},q]e^{\frac{j2\pi}{N} q\big\lfloor \frac{n}{M} \big\rfloor}.
    \end{align}

    Substituting~\eqref{eq:hw_dd2} in Definition~\ref{def:heis_op}, we obtain:
    \begin{align}
        \label{eq:hw_dd3}
        \big(\mathcal{D}_{(k,l)} \circ \mathcal{Z}^{-1}\mathbf{X}\big)[n] &= \frac{1}{\sqrt{N}} \sum_{q = 0}^{N-1}\mathbf{X}[((n-k)_{{}_{MN}})_{{}_{M}},q]e^{\frac{j2\pi}{N} q\big\lfloor \frac{(n-k)_{{}_{MN}}}{M} \big\rfloor} e^{\frac{j2\pi}{MN}l(n-k)} \nonumber \\
        &= \frac{1}{\sqrt{N}} \sum_{q = 0}^{N-1}\mathbf{X}[(n-k)_{{}_{M}},q]e^{\frac{j2\pi}{N} q\big\lfloor \frac{n-k}{M} \big\rfloor} e^{\frac{j2\pi}{MN}l(n-k)},
    \end{align}
    where the final expression follows from elementary modular arithmetic identities.

    Applying the DZT per Definition~\ref{def:dzt}:
    \begin{align}
        \label{eq:hw_dd4}
        \big(\mathcal{Z} \circ \mathcal{D}_{(k,l)} \circ \mathcal{Z}^{-1}\mathbf{X}\big)[k',l'] &= \frac{1}{\sqrt{N}} \sum_{p = 0}^{N-1}\big(\mathcal{D}_{(k,l)} \circ \mathcal{Z}^{-1}\mathbf{X}\big)[k'+pM]e^{-\frac{j2\pi}{N} pl'} \nonumber \\
        &= \frac{1}{N} \sum_{p = 0}^{N-1} \sum_{q = 0}^{N-1} \mathbf{X}[(k'+pM-k)_{{}_{M}},q]e^{\frac{j2\pi}{N} q\big\lfloor \frac{k'+pM-k}{M} \big\rfloor} \nonumber \\ &~~~~~~~~~~~~~~~~~~\times e^{\frac{j2\pi}{MN}l(k'+pM-k)} e^{-\frac{j2\pi}{N} pl'} \nonumber \\
        &= \sum_{q = 0}^{N-1} \mathbf{X}[(k'-k)_{{}_{M}},q] e^{\frac{j2\pi}{N} q\big\lfloor \frac{k'-k}{M} \big\rfloor} e^{\frac{j2\pi}{MN}l(k'-k)} \nonumber \\ &~~~~~~~\times \bigg(\frac{1}{N} \sum_{p = 0}^{N-1} e^{\frac{j2\pi}{N} (q+l-l') p}\bigg).
    \end{align}

    From Identity~\ref{idty:sumrootsofunity}, the inner summation over $p$ vanishes unless $q \equiv (l'-l) \bmod{N}$. Thus:
    \begin{align}
        \label{eq:hw_dd5}
        \big(\mathcal{Z} \circ \mathcal{D}_{(k,l)} \circ \mathcal{Z}^{-1}\mathbf{X}\big)[k',l'] &= \mathbf{X}[(k'-k)_{{}_{M}},(l'-l)_{{}_{N}}] e^{\frac{j2\pi}{N}(l'-l)\big\lfloor \frac{k'-k}{M} \big\rfloor} e^{\frac{j2\pi}{MN}l(k'-k)} \nonumber \\
        &= \big(\mathcal{D}_{(k,l)}^{\mathsf{DD}} \mathbf{X}\big)[k',l'].
    \end{align}

    Hence, $\mathcal{H}_{MN}^{\mathsf{DD}}$ is unitarily equivalent to $\mathcal{H}_{MN}$ from Theorem~\ref{thm:hw_group}.
\end{proof}

Given the unitary equivalence of $\mathcal{H}_{MN}^{\mathsf{DD}}$ and $\mathcal{H}_{MN}$, in the remainder of this paper we restrict attention to the latter, i.e., the Heisenberg-Weyl group defined over the Hilbert space $\mathcal{H}$ of $MN$-periodic sequences. An equivalent development of all our subsequent results is possible in the delay-Doppler domain based on the results presented above.

\subsection{Commutative Subgroups}
\label{subsec:comm_subgroups}

We now investigate subgroups of the Heisenberg-Weyl group $\mathcal{H}_{MN}$ that are \emph{commutative}, i.e., satisfy the symplectic form in Lemma~\ref{lmm:hw_not_comm}. Commutative subgroups, and especially their specialized forms called \emph{maximal} commutative subgroups that are introduced in Section~\ref{subsec:eigenbasis}, play a special role in the theory of discrete radar since they act on input waveforms simply by a phase multiplication; hence lend desirable properties to ambiguity functions.


\begin{definition}
    \label{def:comm_subgrp}
    \emph{Commutative subgroups} $\mathcal{T}_{MN} \subset \mathcal{H}_{MN}$ of the Heisenberg-Weyl group satisfy the symplectic form in Lemma~\ref{lmm:hw_not_comm}, $l_1 k_2 \equiv l_2 k_1 \bmod{MN}$, for any $t_1 = e^{\frac{j2\pi}{MN}m_1} \mathcal{D}_{(k_1,l_1)} \in \mathcal{T}_{MN}$ and $t_2 = e^{\frac{j2\pi}{MN}m_2} \mathcal{D}_{(k_2,l_2)} \in \mathcal{T}_{MN}$.
\end{definition}


Every commutative subgroup contains the \emph{center} of the group, given by phase-scaled versions of the identity operator:
\begin{align*}
    \mathcal{T}_{MN} = \big\{e^{\frac{j2\pi}{MN}m} \mathcal{D}_{(0,0)} \big| m \in \mathbb{Z}_{MN}\big\} \subset \mathcal{H}_{MN},
\end{align*}
which is easily seen to commute with the whole Heisenberg-Weyl group.

A non-trivial example of a commutative subgroup of the Heisenberg-Weyl group corresponds to a \emph{line} modulo $MN$.


\begin{example}
    \label{ex:comm_subgrp_exgen}
    Consider $\mathcal{T}_{MN} = \big\{e^{\frac{j2\pi}{MN}m} \mathcal{D}_{((x c)_{{}_{MN}},(x d)_{{}_{MN}})} \big| x, m, c, d \in \mathbb{Z}_{MN}, (c,d) = 1\big\}$, which corresponds to a \emph{primitive line} in the delay-Doppler grid generated by the vector $\begin{bmatrix}
        c & d
    \end{bmatrix}^{\top}$.
\end{example}

\begin{specialcase}
    \label{ex:comm_subgrp_ex1}
    Consider $c = M$, $d = N$ in Example~\ref{ex:comm_subgrp_exgen}. Note that for every $x \in \mathbb{Z}_{MN}$, there exist $a \in \mathbb{Z}_{N}$ and $b \in \mathbb{Z}_{M}$ such that $x = (M^{-1}_{{}_{N}}a)M + (N^{-1}_{{}_{M}}b)N$. Substituting the values of $x, c, d$ in Example~\ref{ex:comm_subgrp_exgen} results in:
    \begin{align*}
        \mathcal{T}_{MN} = \big\{e^{\frac{j2\pi}{MN}m} \mathcal{D}_{(aM,bN)} \big| m \in \mathbb{Z}_{MN}, a \in \mathbb{Z}_{{}_{N}}, b \in \mathbb{Z}_{{}_{M}}\big\},
    \end{align*}
    which corresponds to an \emph{$M \times N$ rectangular grid} in the delay-Doppler grid.
\end{specialcase}

\begin{specialcase}
    \label{ex:comm_subgrp_ex2}
    Consider $c, d \in \mathbb{Z}_{MN}$ with $(c,MN) = 1$ and $(d,MN) = 1$ in Example~\ref{ex:comm_subgrp_exgen}. Note that for $(c,MN) = 1$ and $(d,MN) = 1$, we have $k = (x c)_{{}_{MN}}$ and $l = (x d)_{{}_{MN}} = d c^{-1}_{{}_{MN}} k$. Let $d c^{-1}_{{}_{MN}} = 2 \alpha$ for some $\alpha \in \mathbf{Z}_{MN}$. Substituting in Example~\ref{ex:comm_subgrp_exgen} results in:
    \begin{align*}
        \mathcal{T}_{MN} = \big\{e^{\frac{j2\pi}{MN}m} \mathcal{D}_{(k,l)} \big| k,l,m \in \mathbb{Z}_{MN}, 2\alpha k - l \equiv 0 \bmod{MN}, (\alpha,MN) = 1\big\},
    \end{align*}
    which corresponds to a \emph{line with slope $2 \alpha$} in the delay-Doppler grid.
\end{specialcase}

\begin{lemma}
    \label{lmm:comm_subgrp_exgen}
    The set $\mathcal{T}_{MN} \subset \mathcal{H}_{MN}$ defined in Example~\ref{ex:comm_subgrp_exgen} under operator composition (equivalently, matrix multiplication) forms a commutative subgroup of $\mathcal{H}_{MN}$.
\end{lemma}

\begin{proof}

    We show that $\mathcal{T}_{MN}$ is a commutative subgroup of $\mathcal{H}_{MN}$ per Definition~\ref{def:comm_subgrp}. Note that $\mathcal{T}_{MN}$ contains the identity element $\mathcal{D}_{(0,0)} \in \mathcal{H}_{MN}$. Hence, it suffices to show commutativity and closure, and the other group properties (associativity and inverse) hold by default.

    \textbf{(Commutativity)} Consider elements $t_1 = e^{\frac{j2\pi}{MN}m_1} \mathcal{D}_{((x_1 c)_{{}_{MN}},(x_1 d)_{{}_{MN}})} \in \mathcal{T}_{MN}$ and $t_2 = e^{\frac{j2\pi}{MN}m_2} \mathcal{D}_{((x_2 c)_{{}_{MN}},(x_2 d)_{{}_{MN}})} \in \mathcal{T}_{MN}$. The symplectic form in Lemma~\ref{lmm:hw_not_comm} and Definition~\ref{def:comm_subgrp} is satisfied by default since:
    \begin{align}
        \label{eq:exgen1}
        l_1k_2 = x_1 d x_2 c \bmod{MN} \equiv l_2k_1 = x_2 d x_1 c \bmod{MN}.
    \end{align}

    
    

    \textbf{(Closure)} The subgroup is closed since on substituting~\eqref{eq:exgen1} in~\eqref{eq:hw3a} we obtain:
    \begin{equation*}
        (t_1 \circ t_2) = (t_2 \circ t_1) = e^{\frac{j2\pi}{MN}(m_1+m_2+x_1 x_2 cd)} \mathcal{D}_{([(x_1+x_2) c]_{{}_{MN}},[(x_1+x_2) d]_{{}_{MN}})} \in \mathcal{T}_{MN}.
    \end{equation*}
\end{proof}

\subsection{Maximal Commutative Subgroups}
\label{subsec:eigenbasis}

We now define the notion of \emph{maximal commutative subgroups} of $\mathcal{H}_{MN}$, which are special commutative subgroups with the property that no group element outside of the subgroup commutes with the subgroup.

\begin{definition}
    \label{def:max_comm_subgrp}
    A commutative subgroup $\mathcal{T}_{MN} \subset \mathcal{H}_{MN}$ is \emph{maximal} if and only if no $h \in \mathcal{H}_{MN}, h \not\in \mathcal{T}_{MN}$ exists such that $h$ commutes with $\mathcal{T}_{MN}$.
\end{definition}


Maximal commutative subgroups of $\mathcal{H}_{MN}$ admit an orthonormal eigenbasis decomposition with unimodular eigenvalues, as described in the following Lemma.

\begin{lemma}
    \label{lmm:eigenbasis}
    For each maximal commutative subgroup $\mathcal{T}_{MN} \subset \mathcal{H}_{MN}$, there exists a simultaneous orthonormal eigenbasis $\big\{\mathbf{v}_{i}\big\}_{i=1}^{MN}$ for all elements $t \in \mathcal{T}_{MN}$, with group elements $h \not\in \mathcal{T}_{MN}$ not in the subgroup acting orthogonally on the eigenvectors: 
    \begin{gather*}
        t \mathbf{v}_{i} = \lambda_{t,i} \mathbf{v}_{i},~t \in \mathcal{T}_{MN}, |\lambda_{t,i}| = 1, \\
        \big\langle \mathbf{v}_{i},h \mathbf{v}_{i} \big\rangle = 0,~h \not\in \mathcal{T}_{MN}, h \in \mathcal{H}_{MN}.
    \end{gather*}
\end{lemma}

\begin{proof}
    Since the elements of $\mathcal{T}_{MN}$ are unitary and commuting, its follows from linear algebra theory that they must all be simultaneously diagonalizable. Hence, there must exist an orthonormal basis $\big\{\mathbf{v}_{i}\big\}_{i=1}^{MN}$ such that for all $t \in \mathcal{T}_{MN}$, $t \mathbf{v}_{i} = \lambda_{t,i} \mathbf{v}_{i}$. Since the elements of $\mathcal{T}_{MN}$ are unitary, the corresponding eigenvalues must have unit magnitude, $|\lambda_{t,i}| = 1$. 

    We prove the final condition by contradiction. Assume there exists some $h \in \mathcal{H}_{MN}, h \not\in \mathcal{T}_{MN}$ such that $h \circ t = t \circ h$ for $t \in \mathcal{T}_{MN}$. Hence, for an eigenvector $\mathbf{v}_{i}$, we must have:
    \begin{align*}
        &h\big(t\mathbf{v}_{i}\big) = t\big(h\mathbf{v}_{i}\big) \nonumber \\
        \implies& \lambda_{t,i} \big(h\mathbf{v}_{i}\big) = t\big(h\mathbf{v}_{i}\big),
    \end{align*}
    i.e., $h\mathbf{v}_{i} \propto \mathbf{v}_{i} \implies h \in \mathcal{T}_{MN}$. However, this is a contradiction since we assumed $h \not\in \mathcal{T}_{MN}$.
    
    Therefore, for all group elements $h \in \mathcal{H}_{MN}, h \not\in \mathcal{T}_{MN}$, we must have $h\mathbf{v}_{i} \propto \mathbf{v}_{j}$ for some $j \neq i$ such that $\lambda_{t,i} = \lambda_{t,j}$. Let $h\mathbf{v}_{i} = \alpha_{i,j} \mathbf{v}_{j}$ for some $\alpha_{i,j} \in \mathbb{C}$. Thus, we have:
    \begin{align*}
        \big\langle \mathbf{v}_{i},h \mathbf{v}_{i} \big\rangle &= \alpha_{i,j}^{*} \big\langle \mathbf{v}_{i}, \mathbf{v}_{j} \big\rangle \nonumber \\
        &= 0,
    \end{align*}
    since the eigenbasis is orthonormal and $j \neq i$.
\end{proof}

Given an eigenvector $\mathbf{v}$ of a maximal commutative subgroup, its self-ambiguity function $\mathbf{A}_{\mathbf{v}}[k,l]$ depends on the orbit of $\mathbf{v}$ under $\mathcal{H}_{MN}$~\cite{Howard2006_HW}. The \emph{isotropy group} $\mathcal{I}_{\mathbf{v}}$ of $\mathbf{v}$ comprises the elements of $\mathcal{H}_{MN}$ that act as multiplication by a complex phase. Isotropy groups are commutative, hence $\mathbf{v}$ is a waveform with the largest possible isotropy group. Corollary~\ref{corr:max_comm_ambg} shows that the self-ambiguity function $\mathbf{A}_{\mathbf{v}}[k,l]$ is constant on cosets of $\mathcal{I}_{\mathbf{v}}$.



\begin{corollary}
    \label{corr:max_comm_ambg}
    The self-ambiguity functions of eigenvectors $\big\{\mathbf{v}_{i}\big\}_{i=1}^{MN}$ of a maximal commutative subgroup $\mathcal{T}_{MN} \subset \mathcal{H}_{MN}$ satisfy:
    \begin{align*}
        \big|\mathbf{A}_{\mathbf{v}_{i}}[k, l]\big| = \begin{cases}
            1, & t = e^{\frac{j2\pi}{MN}m} \mathcal{D}_{(k,l)} \in \mathcal{T}_{MN}, \\ 
            0, & h = e^{\frac{j2\pi}{MN}m} \mathcal{D}_{(k,l)} \not\in \mathcal{T}_{MN}, h \in \mathcal{H}_{MN}.
        \end{cases}
    \end{align*}

    Moreover, the self-ambiguity function is unimodular over exactly $MN$ delay-Doppler indices $(k,l)$, $k \in \mathbb{Z}_{MN}$, $l \in \mathbb{Z}_{MN}$.

\end{corollary}

\begin{proof}
    Substituting Lemma~\ref{lmm:eigenbasis} in~\eqref{eq:amb_fun_hw} yields the desired result. A direct application of Moyal's Identity (Corollary~\ref{corr:moyal_setcard}) results in the final condition. 
\end{proof}


    

Lemma~\ref{lmm:eigenbasis} and Corollary~\ref{corr:max_comm_ambg} show that the self-ambiguity functions of eigenvectors of maximal commutative subgroups are unimodular only over $MN$ delay-Doppler locations, outside of which they are zero. Around each such location, the ambiguity function has a local ``thumbtack-like'' property, making it desirable to utilize eigenvectors of maximal commutative subgroups as radar waveforms.

We now show that Example~\ref{ex:comm_subgrp_exgen} corresponds to a maximal commutative subgroup.

\begin{lemma}
    \label{lmm:max_comm_subgrp_exgen}
    The commutative subgroup $\mathcal{T}_{MN}$ in Example~\ref{ex:comm_subgrp_exgen} is maximal. Moreover, every maximal commutative subgroup of $\mathcal{H}_{MN}$ is a line of the form given by $\mathcal{T}_{MN}$. 
\end{lemma}

\begin{proof}
    Let $\mathcal{S}_{\mathcal{T}_{MN}} = \big\{(k,l) \big| e^{\frac{j2\pi}{MN}m} \mathcal{D}_{(k,l)} \in \mathcal{T}_{MN} \big\}$. A simple counting argument suggests that $\mathsf{card}\big(\mathcal{S}_{\mathcal{T}_{MN}}\big) = MN$ since the delay and Doppler indices $k$ and $l$ are uniquely determined by $x \in \mathbb{Z}_{MN}$ in Example~\ref{ex:comm_subgrp_exgen}. Hence, the condition in Corollary~\ref{corr:max_comm_ambg} is satisfied, i.e., $\mathcal{T}_{MN}$ is a maximal commutative subgroup.

    To show the final condition, consider the quotient group $\mathcal{Q}_{MN} = \mathcal{H}_{MN} / \mathcal{C}_{MN}$, where $\mathcal{C}_{MN} = \big\{e^{\frac{j2\pi}{MN}m} \mathcal{D}_{(0,0)} \big| m \in \mathbb{Z}_{MN}\big\}$ denotes the center of the group. From Lemma~\ref{lmm:hw_not_comm}, it is easy to see that the quotient group is commutative. A maximal commutative subgroup of $\mathcal{H}_{MN}$ corresponds to a (phase-scaled) subgroup of $\mathcal{Q}_{MN}$ with order $MN$. Since $M$ and $N$ are distinct primes, such a subgroup of $\mathcal{Q}_{MN}$ should be cyclic of order $MN$; hence, must be of the form given by $\mathcal{T}_{MN}$.
\end{proof}


We now derive the eigenbases corresponding to the two special cases of Example~\ref{ex:comm_subgrp_exgen}.

\subsubsection{Special Case~\ref{ex:comm_subgrp_ex1}} 

To understand the structure of the eigenbasis, we substitute $\mathcal{T}_{MN}$ from Special Case~\ref{ex:comm_subgrp_ex1} in Lemma~\ref{lmm:eigenbasis}:
\begin{align}
    \label{eq:eigenvec_amb_ex1_1}
    (t \mathbf{v}_{i})[n] &= e^{\frac{j2\pi}{MN}m} \big(\mathcal{D}_{(aM,bN)} \mathbf{v}_{i}\big)[n] \nonumber \\ 
    &= e^{\frac{j2\pi}{MN}m} e^{\frac{j2\pi}{M}bn} \mathbf{v}_{i}[n-aM] = \lambda_{t,i} \mathbf{v}_{i}[n].
\end{align}

For the above condition to be satisfied, we seek an orthonormal basis that is \emph{quasi-periodic} with period $M$ along delay.

\begin{lemma}
    \label{lmm:eigenvec_amb_ex1_pulsone}
    The \emph{pulsone basis} satisfies the condition in~\eqref{eq:eigenvec_amb_ex1_1}, i.e., is the eigenbasis corresponding to the maximal commutative subgroup $\mathcal{T}_{MN}$ in Special Case~\ref{ex:comm_subgrp_ex1}:
    \begin{align*}
        \mathbf{v}_{i}[n] = \frac{1}{\sqrt{N}} \sum_{d \in \mathbb{Z}} e^{\frac{j2\pi}{N} d l_0} \delta[n-k_0-dM],~i=k_0+l_0M,
    \end{align*}
    with corresponding eigenvalue $\lambda_{t,i} = e^{\frac{j2\pi}{MN}m} e^{\frac{j2\pi}{M}bk_0} e^{-\frac{j2\pi}{N} al_0}$.
\end{lemma}

\begin{proof}
    To prove the Lemma, we first evaluate $(h \mathbf{v}_{i})[n]$ for a Heisenberg-Weyl group element $h = e^{\frac{j2\pi}{MN}m} \mathcal{D}_{(k,l)} \in \mathcal{H}_{MN}$:
    \begin{align*}
        (h \mathbf{v}_{i})[n] &= e^{\frac{j2\pi}{MN}m} \big(\mathcal{D}_{(k,l)} \mathbf{v}_{i}\big)[n] \nonumber \\
        &= e^{\frac{j2\pi}{MN}m} e^{\frac{j2\pi}{MN}l(n-k)} \frac{1}{\sqrt{N}} \sum_{d \in \mathbb{Z}} e^{\frac{j2\pi}{N} d l_0} \delta[n-k-k_0-dM],
    \end{align*}
    where the final expression follows on substituting the pulsone basis element defined in the Lemma statement in Definition~\ref{def:heis_op}. On substituting $k + k_0 = (k + k_0)_{{}_{M}} + \big\lfloor\frac{k + k_0}{M}\big\rfloor M$:
    \begin{align*}
        (h \mathbf{v}_{i})[n] &= e^{\frac{j2\pi}{MN}m} e^{\frac{j2\pi}{MN}l(n-k)} \frac{1}{\sqrt{N}} \sum_{d \in \mathbb{Z}} e^{\frac{j2\pi}{N} d l_0} \delta\bigg[n-(k+k_0)_{{}_{M}}-\bigg(d+\bigg\lfloor\frac{k + k_0}{M}\bigg\rfloor\bigg)M\bigg] \nonumber \\
        &= e^{\frac{j2\pi}{MN}m} \frac{1}{\sqrt{N}} \sum_{d \in \mathbb{Z}} e^{\frac{j2\pi}{MN}l\big[(k+k_0)_{{}_{M}}+\big(d+\big\lfloor\frac{k + k_0}{M}\big\rfloor\big)M-k\big]} e^{\frac{j2\pi}{N} d l_0} \nonumber \\ &~~~\times \delta\bigg[n-(k+k_0)_{{}_{M}}-\bigg(d+\bigg\lfloor\frac{k + k_0}{M}\bigg\rfloor\bigg)M\bigg],
    \end{align*}
    where the final expression follows from the property of the Kronecker delta function that $f[n] \delta[n-m] = f[m] \delta[n-m]$. Let $d' = d+\big\lfloor\frac{k + k_0}{M}\big\rfloor$. We have:
    \begin{align}
        \label{eq:eigenvec_amb_ex1_pulsone3}
        (h \mathbf{v}_{i})[n] &= e^{\frac{j2\pi}{MN}m} e^{-\frac{j2\pi}{N} \big\lfloor\frac{k + k_0}{M}\big\rfloor l_0} \frac{1}{\sqrt{N}} \sum_{d' \in \mathbb{Z}} e^{\frac{j2\pi}{MN}l[(k+k_0)_{{}_{M}}+d'M-k]} e^{\frac{j2\pi}{N} d' l_0} \nonumber \\ &~~~\times \delta[n-(k+k_0)_{{}_{M}}-d'M] \nonumber \\
        &= e^{\frac{j2\pi}{MN}m} e^{\frac{j2\pi}{MN}l[(k+k_0)_{{}_{M}}-k]} e^{-\frac{j2\pi}{N} \big\lfloor\frac{k + k_0}{M}\big\rfloor l_0} \frac{1}{\sqrt{N}} \sum_{d' \in \mathbb{Z}} e^{\frac{j2\pi}{N} d' (l+l_0)} \nonumber \\ &~~~\times \delta[n-(k+k_0)_{{}_{M}}-d'M] \nonumber \\
        &= e^{\frac{j2\pi}{MN}m} e^{\frac{j2\pi}{MN}l[(k+k_0)_{{}_{M}}-k]} e^{-\frac{j2\pi}{N} \big\lfloor\frac{k + k_0}{M}\big\rfloor l_0} \mathbf{v}_{j}[n],~j = (k+k_0)_{{}_{M}} + (l+l_0)_{{}_{N}}M.
    \end{align}

    
    When $h = t \in \mathcal{T}_{MN}$, i.e., $k = aM$ and $l = bN$, we have $j = i = k_0 + l_0 M$; hence:
    \begin{align}
        \label{eq:eigenvec_amb_ex1_pulsone5}
        t \mathbf{v}_{i} &= e^{\frac{j2\pi}{MN}m} e^{\frac{j2\pi}{MN}l[(k+k_0)_{{}_{M}}-k]} e^{-\frac{j2\pi}{N} \big\lfloor\frac{k + k_0}{M}\big\rfloor l_0} \mathbf{v}_{i} \nonumber \\
        &= e^{\frac{j2\pi}{MN}m} e^{\frac{j2\pi}{M}b[(aM+k_0)_{{}_{M}}-aM]} e^{-\frac{j2\pi}{N} \big\lfloor\frac{aM + k_0}{M}\big\rfloor l_0} \mathbf{v}_{i} \nonumber \\
        &= \underbrace{e^{\frac{j2\pi}{MN}m} e^{\frac{j2\pi}{M}bk_0} e^{-\frac{j2\pi}{N} al_0}}_{\lambda_{t,i}} \mathbf{v}_{i}.
    \end{align}

    When $h \not\in \mathcal{T}_{MN}$, i.e., $k \not\equiv 0 \bmod{M}$ and/or $l \not\equiv 0 \bmod{N}$,~\eqref{eq:eigenvec_amb_ex1_pulsone3} implies $j \neq i$; hence:
    \begin{align}
        \label{eq:eigenvec_amb_ex1_pulsone4}
        \big\langle \mathbf{v}_{i}, h \mathbf{v}_{i} \big\rangle &= e^{-\frac{j2\pi}{MN}m} e^{-\frac{j2\pi}{MN}l[(k+k_0)_{{}_{M}}-k]} e^{\frac{j2\pi}{N} \big\lfloor\frac{k + k_0}{M}\big\rfloor l_0} \big\langle \mathbf{v}_{i}, \mathbf{v}_{j} \big\rangle \nonumber \\
        &= 0,
    \end{align}
    since the pulsone basis is orthonormal. Together,~\eqref{eq:eigenvec_amb_ex1_pulsone4} and~\eqref{eq:eigenvec_amb_ex1_pulsone5} satisfy the conditions in Definition~\ref{def:max_comm_subgrp} and Corollary~\ref{corr:max_comm_ambg}. 
\end{proof}

\subsubsection{Special Case~\ref{ex:comm_subgrp_ex2}} 

To understand the structure of the eigenbasis, we substitute $\mathcal{T}_{MN}$ from Special Case~\ref{ex:comm_subgrp_ex2} in Lemma~\ref{lmm:eigenbasis}:
\begin{align}
    \label{eq:eigenvec_amb_ex2_1}
    (t \mathbf{v}_{i})[n] &= e^{\frac{j2\pi}{MN}m} \big(\mathcal{D}_{(k,2\alpha k)} \mathbf{v}_{i}\big)[n] \nonumber \\ 
    &= e^{\frac{j2\pi}{MN}m} e^{\frac{j2\pi}{MN}2\alpha k(n-k)} \mathbf{v}_{i}[n-k] = \lambda_{t,i} \mathbf{v}_{i}[n].
\end{align}

For the above condition to be satisfied, we seek an orthonormal basis with \emph{quadratic phase terms}.

\begin{lemma}
    \label{lmm:eigenvec_amb_ex2_tdcazac}
    The \emph{discrete chirp basis} satisfies the conditions in~\eqref{eq:eigenvec_amb_ex2_1}, i.e., is the eigenbasis corresponding to the maximal commutative subgroup $\mathcal{T}_{MN}$ in Special Case~\ref{ex:comm_subgrp_ex2}:
    \begin{align*}
        \mathbf{v}_{i}[n] = \frac{1}{\sqrt{MN}} e^{\frac{j2\pi}{MN} [\alpha n^2 + \beta n + \gamma]},~i = f(\alpha,\beta,\gamma),
    \end{align*}
    for an affine function $f$, with corresponding eigenvalue $\lambda_{t,i} = e^{\frac{j2\pi}{MN}m} e^{-\frac{j2\pi}{MN}(\alpha k^2 + \beta k)}$. 
\end{lemma}

\begin{proof}
    To prove the Lemma, we first evaluate $(h \mathbf{v}_{i})[n]$ for a Heisenberg-Weyl group element $h = e^{\frac{j2\pi}{MN}m} \mathcal{D}_{(k,l)} \in \mathcal{H}_{MN}$:
    \begin{align}
        \label{eq:eigenvec_amb_ex2_tdcazac1}
        (h \mathbf{v}_{i})[n] &= e^{\frac{j2\pi}{MN}m} \big(\mathcal{D}_{(k,l)} \mathbf{v}_{i}\big)[n] \nonumber \\
        &= e^{\frac{j2\pi}{MN}m} e^{\frac{j2\pi}{MN}l(n-k)} \frac{1}{\sqrt{MN}} e^{\frac{j2\pi}{MN} [\alpha (n-k)^2 + \beta (n-k) + \gamma]} \nonumber \\
        &= e^{\frac{j2\pi}{MN}m} e^{\frac{j2\pi}{MN}[\alpha k^2 - \beta k - lk]} \frac{1}{\sqrt{MN}} e^{\frac{j2\pi}{MN} [\alpha n^2 + (\beta + l - 2\alpha k) n + \gamma]} \nonumber \\
        &= e^{\frac{j2\pi}{MN}m} e^{\frac{j2\pi}{MN}[\alpha k^2 - \beta k - lk]} \mathbf{v}_{j}[n],~j = f(\alpha,\beta + l - 2\alpha k,\gamma).
    \end{align}

    When $h = t \in \mathcal{T}_{MN}$, i.e., $2\alpha k - l \equiv 0 \bmod{MN}$, we have $j = i = f(\alpha,\beta,\gamma)$; hence:
    \begin{align}
        \label{eq:eigenvec_amb_ex2_tdcazac2}
        t \mathbf{v}_{i} &= e^{\frac{j2\pi}{MN}m} e^{\frac{j2\pi}{MN}[\alpha k^2 - \beta k - lk]} \mathds{1}\big\{2\alpha k - l \equiv 0 \bmod{MN}\big\} \mathbf{v}_{i} \nonumber \\
        &= \underbrace{e^{\frac{j2\pi}{MN}m} e^{-\frac{j2\pi}{MN}(\alpha k^2 + \beta k)}}_{\lambda_{t,i}} \mathbf{v}_{i}. 
    \end{align}

    Following~\eqref{eq:eigenvec_amb_ex2_tdcazac1}, we may evaluate the inner product with $\mathbf{v}_{i}$ as:
    \begin{align}
        \label{eq:eigenvec_amb_ex2_tdcazac3}
        \big\langle \mathbf{v}_{i}, h \mathbf{v}_{i} \big\rangle &= e^{-\frac{j2\pi}{MN}m} e^{-\frac{j2\pi}{MN}[\alpha k^2 - \beta k - lk]} \big\langle \mathbf{v}_{i}, \mathbf{v}_{j} \big\rangle \nonumber \\
        &= e^{-\frac{j2\pi}{MN}m} e^{-\frac{j2\pi}{MN}[\alpha k^2 - \beta k - lk]} \frac{1}{MN} \sum_{n=0}^{MN-1} e^{\frac{j2\pi}{MN} (2\alpha k-l) n} \nonumber \\
        &= e^{-\frac{j2\pi}{MN}m} e^{-\frac{j2\pi}{MN}[\alpha k^2 - \beta k - lk]} \mathds{1}\big\{2\alpha k - l \equiv 0 \bmod{MN}\big\},
    \end{align}
    where the final expression follows from Identity~\ref{idty:sumrootsofunity}. Thus,~\eqref{eq:eigenvec_amb_ex2_tdcazac3} satisfies the conditions in Definition~\ref{def:max_comm_subgrp} and Corollary~\ref{corr:max_comm_ambg}. 
\end{proof}

\subsection{Symplectic Transformations}
\label{subsec:gdaft}

We now define a special class of unitary operators called \emph{symplectic transformations} that act by conjugation on $\mathcal{H}_{MN}$, defining automorphisms that preserve the group structure and hence map commutative subgroups to commutative subgroups.


\begin{definition}
    \label{def:sympl_transform}
    A \emph{symplectic transformation} of the Heisenberg-Weyl group $\mathcal{H}_{MN}$ is a unitary operator $\mathcal{W}(g): \mathcal{H} \rightarrow \mathcal{H}$ associated with an element $g \in \mathcal{G}$ of the special linear group of order $2$ over $\mathbb{Z}_{MN}$:
    \begin{align}
        \label{eq:weil1}
        \mathcal{G} &= SL_{2}(\mathbb{Z}_{MN}) = \bigg\{\begin{bmatrix}
            a & b \\ c & d
        \end{bmatrix} \bigg| a, b, c, d \in \mathbb{Z}_{MN}, ad - bc = 1 \bigg\},
    \end{align}
    that acts on the Hilbert space $\mathcal{H}$ of complex-valued $MN$-periodic sequences as:
    \begin{gather}
        \label{eq:weil3a}
        \mathcal{W}(g \cdot h) = \mathcal{W}(g) \circ \mathcal{W}(h), \forall g,h \in \mathcal{G}, \\
        \label{eq:weil3b}
        \mathcal{W}(g) \circ \big(e^{\frac{j2\pi}{MN}m}\mathcal{D}_{(k,l)}\big) \circ \mathcal{W}^{-1}(g) = e^{\frac{j2\pi}{MN}m'}\mathcal{D}_{g \cdot (k,l)}, \forall g \in \mathcal{G}, 
    \end{gather}
    where $e^{\frac{j2\pi}{MN}m}\mathcal{D}_{(k,l)}$ and $e^{\frac{j2\pi}{MN}m'}\mathcal{D}_{g \cdot (k,l)}$ belong to the Heisenberg-Weyl group $\mathcal{H}_{MN}$.
\end{definition}

We may interpret~\eqref{eq:weil3b} as the statement that $\mathcal{W}(g)$ normalizes $\mathcal{H}_{MN}$. Conjugation by $\mathcal{W}^{-1}(g)$ defines an automorphism of $\mathcal{H}_{MN}$ that maps commutative subgroups to commutative subgroups.

\begin{lemma}
    \label{lmm:weil_prop1}
    Symplectic transformations map commutative subgroups of $\mathcal{H}_{MN}$ to commutative subgroups.
\end{lemma}




We now make Lemma~\ref{lmm:weil_prop1} explicit. Consider two elements $t_1 = e^{\frac{j2\pi}{MN}m_1}\mathcal{D}_{(k_1,l_1)} \in \mathcal{T}_{MN}$ and $t_2 = e^{\frac{j2\pi}{MN}m_2}\mathcal{D}_{(k_2,l_2)} \in \mathcal{T}_{MN}$. As per~\eqref{eq:weil3b}:
\begin{align}
    \label{eq:weil4}
    \mathcal{W}(g) \circ \big(e^{\frac{j2\pi}{MN}m_i}\mathcal{D}_{(k_i,l_i)}\big) \circ \mathcal{W}^{-1}(g) = e^{\frac{j2\pi}{MN}m_i'}\mathcal{D}_{g \cdot (k_i,l_i)},\forall i \in \{1,2\}.
\end{align}

Recall from Definition~\ref{def:comm_subgrp} that $t_1,t_2 \in \mathcal{T}_{MN}$ satisfy the symplectic form:
\begin{align}
    \label{eq:weil5}
    l_1 k_2 \equiv l_2 k_1 \bmod{MN} \implies \begin{bmatrix}
        k_1 & l_1
    \end{bmatrix} \begin{bmatrix}
        0 & -1 \\ 1 & 0
    \end{bmatrix} \begin{bmatrix}
        k_2 \\ l_2
    \end{bmatrix} \equiv 0 \bmod{MN}.
\end{align}

Now, for $i \in \{1,2\}$, let $(k'_i,l'_i) = g \cdot (k_i,l_i)$; thus, $(k_i,l_i) = g^{-1} \cdot (k'_i,l'_i)$. Then,~\eqref{eq:weil5} can be rewritten as:
\begin{align}
    \label{eq:weil6}
    \begin{bmatrix}
        k'_1 & l'_1
    \end{bmatrix} \big(g^{-1}\big)^{\top} \begin{bmatrix}
        0 & -1 \\ 1 & 0
    \end{bmatrix} g^{-1} \begin{bmatrix}
        k'_2 \\ l'_2
    \end{bmatrix} \equiv 0 \bmod{MN}.
\end{align}

Substituting the definition of the group element $g \in \mathcal{G}$ as per~\eqref{eq:weil1}, we obtain:
\begin{align}
    \label{eq:weil7}
    \begin{bmatrix}
        k'_1 & l'_1
    \end{bmatrix} \big(g^{-1}\big)^{\top} \begin{bmatrix}
        0 & -1 \\ 1 & 0
    \end{bmatrix} g^{-1} \begin{bmatrix}
        k'_2 \\ l'_2
    \end{bmatrix} &= \begin{bmatrix}
        k'_1 & l'_1
    \end{bmatrix} \bigg(\begin{bmatrix}
        a & b \\ c & d
    \end{bmatrix}^{-1}\bigg)^{\top} \begin{bmatrix}
        0 & -1 \\ 1 & 0
    \end{bmatrix} \begin{bmatrix}
        a & b \\ c & d
    \end{bmatrix}^{-1} \begin{bmatrix}
        k'_2 \\ l'_2
    \end{bmatrix} \nonumber \\
    &= \begin{bmatrix}
        k'_1 & l'_1
    \end{bmatrix} \begin{bmatrix}
        d & -c \\ -b & a
    \end{bmatrix} \begin{bmatrix}
        0 & -1 \\ 1 & 0
    \end{bmatrix} \begin{bmatrix}
        d & -b \\ -c & a
    \end{bmatrix} \begin{bmatrix}
        k'_2 \\ l'_2
    \end{bmatrix} \nonumber \\
    &= \begin{bmatrix}
        k'_1 & l'_1
    \end{bmatrix} \begin{bmatrix}
        d & -c \\ -b & a
    \end{bmatrix} \begin{bmatrix}
        c & -a \\ d & -b
    \end{bmatrix} \begin{bmatrix}
        k'_2 \\ l'_2
    \end{bmatrix} \nonumber \\
    &= \begin{bmatrix}
        k'_1 & l'_1
    \end{bmatrix} \begin{bmatrix}
        (dc-dc) & -(ad-bc) \\ (ad-bc) & (ab-ab)
    \end{bmatrix} \begin{bmatrix}
        k'_2 \\ l'_2
    \end{bmatrix} \nonumber \\
    &= \begin{bmatrix}
        k'_1 & l'_1
    \end{bmatrix} \begin{bmatrix}
        0 & -1 \\ 1 & 0
    \end{bmatrix} \begin{bmatrix}
        k'_2 \\ l'_2
    \end{bmatrix} \equiv 0 \bmod{MN}.
\end{align}
Thus, the transformation $(k'_i,l'_i) = g \cdot (k_i,l_i)$ due to the symplectic transformation $\mathcal{W}(g)$ preserves the symplectic form from Definition~\ref{def:comm_subgrp}, and hence also represents a commutative subgroup of $\mathcal{H}_{MN}$.

\begin{lemma}
    \label{lmm:weil_prop2}
    Symplectic transformations \emph{rotate} the ambiguity functions of waveforms. Let $\mathcal{W}(g)$ denote a symplectic transformation corresponding to a group element $g \in \mathcal{G}$. Then: 
    \begin{align*}
        \mathbf{A}_{\mathbf{x},\mathbf{y}}[k,l] &= \big\langle \mathbf{x}, \mathcal{D}_{(k,l)} \mathbf{y} \big\rangle \nonumber \\
        &= e^{-\frac{j2\pi}{MN}m'} \mathbf{A}_{\mathcal{W}(g) \mathbf{x},\mathcal{W}(g) \mathbf{y}}[g \cdot (k,l)].
    \end{align*}
\end{lemma}


\begin{proof}
    Since $\mathcal{W}(g)$ is unitary, we may express the ambiguity function as:
    \begin{align}
        \label{eq:amb2}
        \mathbf{A}_{\mathbf{x},\mathbf{y}}[k,l] &= \big\langle \mathbf{x}, \mathcal{D}_{(k,l)} \mathbf{y} \big\rangle \nonumber \\
        &= \big\langle \mathcal{W}^{-1}(g) \circ \mathcal{W}(g) \mathbf{x}, \mathcal{W}^{-1}(g) \circ \mathcal{W}(g) \circ \mathcal{D}_{(k,l)} \circ \mathcal{W}^{-1}(g) \circ \mathcal{W}(g) \mathbf{y} \big\rangle \nonumber \\
        &= \big\langle \mathcal{W}^{-1}(g) \circ \mathcal{W}(g) \mathbf{x}, \mathcal{W}^{-1}(g) \circ \big(e^{\frac{j2\pi}{MN}m'}\mathcal{D}_{g \cdot (k,l)}\big) \circ \mathcal{W}(g) \mathbf{y} \big\rangle \nonumber \\
        &= \big\langle \mathcal{W}(g) \mathbf{x}, \big(e^{\frac{j2\pi}{MN}m'}\mathcal{D}_{g \cdot (k,l)}\big) \circ \mathcal{W}(g) \mathbf{y} \big\rangle \nonumber \\
        &= e^{-\frac{j2\pi}{MN}m'} \mathbf{A}_{\mathcal{W}(g) \mathbf{x},\mathcal{W}(g) \mathbf{y}}[g \cdot (k,l)].
    \end{align}
\end{proof}

Fig.~\ref{fig:weil_ambg_fun} below illustrates Lemma~\ref{lmm:weil_prop2}.

\begin{figure}[h!]
  \centering



  \begin{tikzpicture}[>=Latex]
    \node (eqA) at (0, 2) {$\mathbf{A}_{\mathbf{x}}[k,l]$};
    \node (eqB) at (6, 2) {$e^{-\frac{j2\pi}{MN}m'} \mathbf{A}_{\mathcal{W}(g) \mathbf{x}}[g \cdot (k,l)]$};
    \node (eqC) at (0, 0) {$\mathbf{x}[n]$};
    \node (eqD) at (5, 0) {$\big(\mathcal{W}(g) \mathbf{x}\big)[n]$};

    \coordinate (Aleft) at ($(eqA.east)+(0.05,0)$);
    \coordinate (Bleft) at ($(eqB.west)-(0.2,0)$);
    \coordinate (Cleft) at ($(eqC.east)+(0.32,0)$);
    \coordinate (Dleft) at ($(eqD.west)-(0.22,0)$);

    \draw[->] (Aleft) -- (Bleft);
    \draw[->] (Cleft) -- (Dleft);

    \draw[-] ($(eqA.south)+(0,0)$) -- ($(eqC.north)+(0,0)$);
    \draw[-] ($(eqB.south)-(1,0)$) -- ($(eqD.north)-(0,0)$);
  \end{tikzpicture}
  \caption{The action of symplectic transformations on waveforms \& ambiguity functions.}
  \label{fig:weil_ambg_fun}
\end{figure}

An example of a symplectic transformation is the DZT introduced in Definition~\ref{def:dzt}. We showed in Theorem~\ref{thm:hw_group_dd} that the DZT \emph{centralizes} $\mathcal{H}_{MN}$, i.e., the DZT is a symplectic transformation $\mathcal{W}(g)$ corresponding to the identity group element $g = \begin{bmatrix}
        1 & 0 \\ 0 & 1
    \end{bmatrix} \in \mathcal{G}$ with:
\begin{align}
    \label{eq:weil_dzt1}
    \mathbf{A}_{\mathbf{x},\mathbf{y}}[k,l] &= \mathbf{A}_{\mathcal{W}(g) \mathbf{x},\mathcal{W}(g) \mathbf{y}}[k,l],
\end{align}
where the above expression follows from~\eqref{eq:hw_dd1} in the proof of Theorem~\ref{thm:hw_group_dd}.

We now present two examples of symplectic transformations corresponding to non-identity elements $g \in \mathcal{G}$. The presented examples have practical implications for radar signal processing, especially in the design of \emph{waveform libraries} for adaptive radars.



\begin{example}
    \label{ex:sympl_tx_ex1}
    Consider the symplectic transformation:
    \begin{align*}
        (\mathcal{W}(g) \mathbf{x})[n] = e^{\frac{j2\pi}{MN}An^2} \mathbf{x}[n],~(A,MN) = 1,
    \end{align*}
    which is defined for the group element $g = \begin{bmatrix}
        1 & 0 \\ 2A & 1
    \end{bmatrix} \in \mathcal{G}$, and corresponds to a \emph{linear frequency modulation}  (LFM) of the input waveform $\mathbf{x}$ at rate $A$.
\end{example}

Example~\ref{ex:sympl_tx_ex1} is the discrete form of the continuous linear frequency modulation transform presented in~\cite{Moran2004_wvf_lib,Moran2009_wvf_lib_survey}. We now derive the group element $g \in \mathcal{G}$ corresponding to the symplectic transformation. To that end, let $\tilde{\mathbf{y}}[n] = (\mathcal{W}(g) \mathbf{y})[n]$. We calculate:
\begin{align}
    \label{eq:weilex1b}
    (\mathcal{W}(g) \circ \mathcal{D}_{(k,l)} \circ \mathcal{W}^{-1}(g) \tilde{\mathbf{y}})[n] &= e^{\frac{j2\pi}{MN}An^2}\big(\mathcal{D}_{(k,l)} (\mathcal{W}^{-1}(g) \tilde{\mathbf{y}})\big)[n] \nonumber \\
    &= e^{\frac{j2\pi}{MN}An^2} e^{\frac{j2\pi}{MN}l(n-k)} e^{-\frac{j2\pi}{MN}A(n-k)^2} \tilde{\mathbf{y}}[(n-k)_{MN}] \nonumber \\
    &= e^{\frac{j2\pi}{MN}[Ak(2n-k)+l(n-k)]} \tilde{\mathbf{y}}[(n-k)_{MN}] \nonumber \\
    &= e^{\frac{j2\pi}{MN}[(2Ak+l)n-(Ak^2+lk)]} \tilde{\mathbf{y}}[(n-k)_{MN}] \nonumber \\
    &= e^{\frac{j2\pi}{MN}Ak^2} e^{\frac{j2\pi}{MN}(2Ak+l)(n-k)} \tilde{\mathbf{y}}[(n-k)_{MN}] \nonumber \\
    &= e^{\frac{j2\pi}{MN}Ak^2} (\mathcal{D}_{(k,(2Ak+l))} \tilde{\mathbf{y}})[n].
\end{align}

On comparing with~\eqref{eq:amb2}:
\begin{equation}
    \label{eq:weilex1c}
    g \cdot (k,l) = \big(k,(2Ak+l)\big) \implies g = \begin{bmatrix}
        1 & 0 \\ 2A & 1
    \end{bmatrix} \in \mathcal{G},
\end{equation}
and:
\begin{align}
    \label{eq:weilex1d}
    \mathbf{A}_{\mathbf{x},\mathbf{y}}[k,l] &= e^{-\frac{j2\pi}{MN}Ak^2} \mathbf{A}_{\mathcal{W}(g) \mathbf{x},\mathcal{W}(g) \mathbf{y}}[k,(2Ak+l)].
\end{align}

\begin{example}
    \label{ex:sympl_tx_ex2}
    Consider the symplectic transformation:
    \begin{align*}
        (\mathcal{W}(g) \mathbf{x})[n] = \frac{1}{\sqrt{MN}} \sum_{n_1 = 0}^{MN-1} e^{j\frac{\pi}{MN} b^{-1}_{{}_{MN}} (d n^2 - 2 nn_1 + a n_1^2)} \mathbf{x}[n_1],~(b,MN) = 1,
    \end{align*}
    which is defined for any group element $g \in \mathcal{G}$ with $(b,MN) = 1$, and corresponds to a discrete generalization of the \emph{affine Fourier transform}, referred to as GDAFT in short.
\end{example}


Example~\ref{ex:sympl_tx_ex2} is the discrete form of the continuous affine Fourier transform presented in~\cite{Moran2004_wvf_lib,Moran2009_wvf_lib_survey}. We now specialize the derivation in~\eqref{eq:amb2} to this transformation. Let $\tilde{\mathbf{y}}[n] = (\mathcal{W}(g) \mathbf{x})[n]$ as before. We calculate:
\begin{align}
    \label{eq:weilex2c}
    (\mathcal{W}(g) \circ \mathcal{D}_{(k,l)} \circ \mathcal{W}^{-1}(g) \tilde{\mathbf{y}})[n] &= \frac{1}{MN} \sum_{n_2 = 0}^{MN-1} e^{j\frac{\pi}{MN} b^{-1}_{{}_{MN}} (d n^2 - 2 nn_2 + a n_2^2)} e^{\frac{j2\pi}{MN}l(n_2-k)} \nonumber \\ &~~~~~~~\times\sum_{n_1 = 0}^{MN-1} e^{-j\frac{\pi}{MN} b^{-1}_{{}_{MN}} (d n_1^2 - 2 n_1(n_2-k) + a (n_2-k)^2)} \tilde{\mathbf{y}}[n_1] \nonumber \\
    &= e^{j\frac{\pi}{MN} [b^{-1}_{{}_{MN}} (d n^2-ak^2)]} \sum_{n_1 = 0}^{MN-1} \tilde{\mathbf{y}}[n_1] e^{-j\frac{\pi}{MN} b^{-1}_{{}_{MN}} (d n_1^2 + 2 n_1k)} \nonumber \\ &~\times e^{-j\frac{2\pi}{MN}lk}\frac{1}{MN} \sum_{n_2 = 0}^{MN-1} e^{j\frac{2\pi}{MN}[l+b^{-1}_{{}_{MN}}(n_1-n+ak)]n_2}.
\end{align}

The inner summation over $n_2$ vanishes unless $l+b^{-1}_{{}_{MN}}(n_1-n+ak) \equiv 0 \bmod{MN} \implies n_1 = (n-(ak+bl))_{{}_{MN}}$, when it takes the value $MN$. Thus:
\begin{align}
    \label{eq:weilex2d}
    (\mathcal{W}(g) \circ \mathcal{D}_{(k,l)} \circ \mathcal{W}^{-1}(g) \tilde{\mathbf{y}})[n] &= e^{j\frac{\pi}{MN} [b^{-1}_{{}_{MN}} (d n^2-ak^2) - 2lk]} \tilde{\mathbf{y}}[(n-(ak+bl))_{{}_{MN}}] \nonumber \\ &~~~\times e^{-j\frac{\pi}{MN} b^{-1}_{{}_{MN}} (d (n-(ak+bl))^2 + 2 (n-(ak+bl))k)}  \nonumber \\
    &= e^{j\frac{\pi}{MN} b^{-1}_{{}_{MN}} [d(ak + bl)^2 - ak^2 - 2blk]} \tilde{\mathbf{y}}[(n-(ak+bl))_{{}_{MN}}] \nonumber \\ &~~~\times e^{j\frac{2\pi}{MN} b^{-1}_{{}_{MN}} [d(ak+bl)-k] (n-(ak+bl))} \nonumber \\
    &= e^{j\frac{\pi}{MN} [ack^2 + bdl^2 + 2bclk]} \nonumber \\ &~~~\times e^{j\frac{2\pi}{MN} [ck+dl] (n-(ak+bl))} \tilde{\mathbf{y}}[(n-(ak+bl))_{{}_{MN}}] \nonumber \\
    &= e^{j\frac{\pi}{MN} [ack^2 + bdl^2 + 2bclk]} (\mathcal{D}_{g\cdot(k,l)} \tilde{\mathbf{y}})[n],
\end{align}
where $g = \begin{bmatrix}
        a & b \\ c & d
    \end{bmatrix} \in \mathcal{G}$. On comparing with~\eqref{eq:amb2}:
\begin{align}
    \label{eq:weilex2e}
    \mathbf{A}_{\mathbf{x},\mathbf{y}}[k,l] &= e^{-j\frac{\pi}{MN} [ack^2 + bdl^2 + 2bclk]} \mathbf{A}_{\mathcal{W}(g) \mathbf{x},\mathcal{W}(g) \mathbf{y}}[g\cdot(k,l)].
\end{align}

\section{A New Framework for Discrete Radar}
\label{sec:impl_radar}


In this Section, we describe a new architecture for discrete radar that implements the Heisenberg-Weyl group theory described in Section~\ref{sec:hw_group}. Figure~\ref{fig:block_diag} illustrates three different radar architectures -- continuous radar, discrete radar with separate optimization of sequences and carrier waveforms, and our proposed approach.

\begin{figure}[!ht]
\centering
\begin{subfigure}{0.95\columnwidth}
    \includegraphics[width=\textwidth]{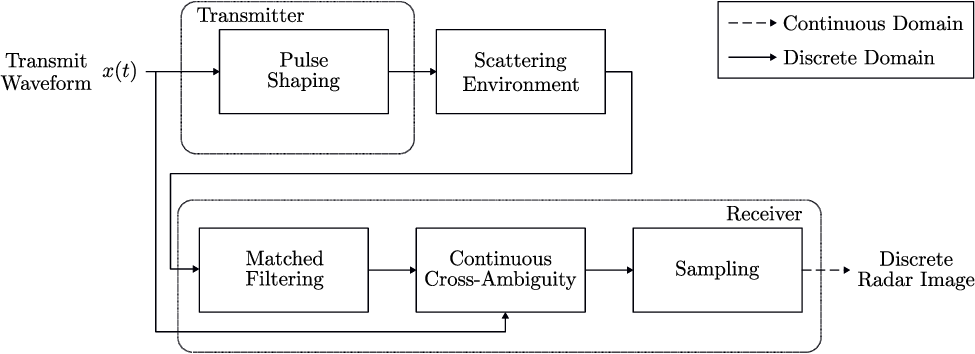}
\caption{Continuous radar architecture.}
    \label{fig:cont_radar}
\end{subfigure}
\begin{subfigure}{0.95\columnwidth}
    \includegraphics[width=\textwidth]{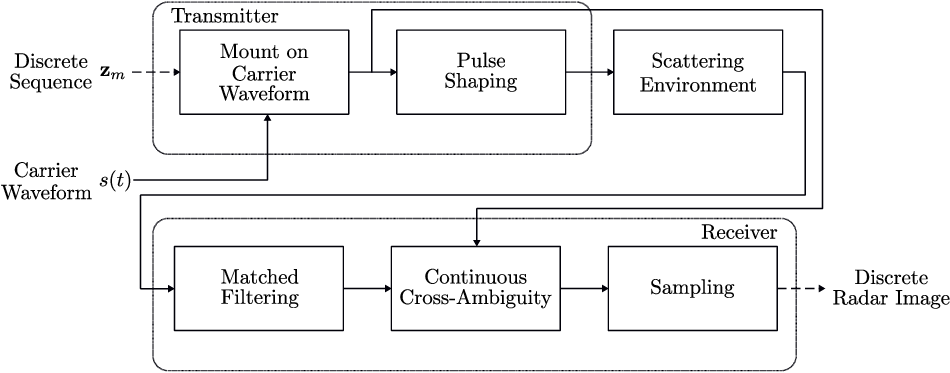}
\caption{Discrete radar architecture with separate optimization of sequences and carriers.}
    \label{fig:disc_phase_coded_radar}
\end{subfigure}
\begin{subfigure}{0.95\columnwidth}
    \includegraphics[width=\textwidth]{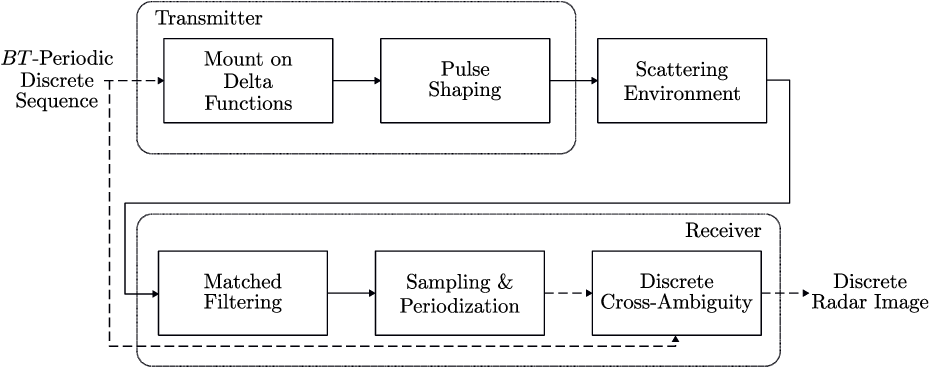}
\caption{Proposed discrete radar architecture.}
    \label{fig:disc_prop_radar}
\end{subfigure}
\caption{Block diagrams showing different radar architectures.}
\vspace{-5mm}
    \label{fig:block_diag}
\end{figure}

Continuous radars~\cite{Auslander1985,Miller1991,Moran2001_mathofradar,Moran2004_grouptheory_radar,Calderbank2015_ltv,Jankiraman2018_fmcw} (Fig.~\ref{fig:block_diag}(\subref{fig:cont_radar})) transmit a continuous TD waveform $x(t)$ after pulse shaping to limit time and bandwidth to $T$ and $B$ respectively. The receiver performs matched filtering to remove the impact of pulse shaping, following which a continuous domain cross-ambiguity function is computed between the output of the matched filter and the transmit waveform. The radar forms a discrete radar image by sampling the continuous cross-ambiguity function at multiples of the delay and Doppler resolution.

Conventional discrete radars~\cite{benedetto_phasecoded,Pezeshki2008,Dang2020,Tang2022,Costas1984,Bell2003,Vehmas2021,Pezeshki2009,Fei2024_phasecoded_compl} (see Fig.~\ref{fig:block_diag}(\subref{fig:disc_phase_coded_radar})) mount a discrete sequence $\mathbf{z}_m$ onto a continuous carrier waveform (e.g., a rectangular pulse train), where the sequence is optimized separately from the carrier waveform. Possible strategies for sequence optimization include coding either the \emph{phase}~\cite{benedetto_phasecoded,Pezeshki2008,Dang2020,Tang2022,Fei2024_phasecoded_compl}, \emph{frequency}~\cite{Costas1984,Bell2003,Vehmas2021} or \emph{amplitude}~\cite{Pezeshki2009} of the carrier waveform. After pulse shaping at the transmitter and matched filtering at the receiver, a continuous cross-ambiguity function is computed, which is sampled to obtain the discrete radar image. 

We propose a different architecture for discrete radar, which is illustrated in Fig.~\ref{fig:block_diag}(\subref{fig:disc_prop_radar}). We mount $BT$-periodic discrete sequences onto delta functions, which are pulse shaped and matched filtered at the received. The output of the matched filter is first sampled and periodized to period $BT$. Then, a discrete cross-ambiguity function is computed between the discrete transmitted sequence and the sampled \& periodized matched filter output to obtain the discrete radar image. The key difference with conventional discrete radars (Fig.~\ref{fig:block_diag}(\subref{fig:disc_phase_coded_radar})) is that our approach jointly optimizes the sequence and carrier waveform, which significantly reduces sidelobes of the ambiguity function and improves target detection, as we show in Section~\ref{subsec:impl_radar_wvf}. The Heisenberg-Weyl group theory described in Section~\ref{sec:hw_group} also enables developing radar waveform libraries with low peak-to-average power ratio (PAPR), which we describe in Section~\ref{subsec:papr_radar}. Moreover, sampling prior to the cross-ambiguity computation results in significant complexity gains over both continuous and discrete radars (Figs.~\ref{fig:block_diag}(\subref{fig:cont_radar})-(\subref{fig:disc_phase_coded_radar})), as we show in Section~\ref{subsec:low_compl_radar}.

Fig.~\ref{fig:heatmaps} compares the discrete radar images obtained using the three radar architectures from Fig.~\ref{fig:block_diag} for a four target scattering environment. In Fig.~\ref{fig:heatmaps}(\subref{fig:heatmaps_fmcw}) we consider a continuous radar with an up-slope and down-slope LFM transmit waveform~\cite{Calderbank2015_ltv}. The resulting radar image suffers from extremely high sidelobes due to the ambiguity function characteristics of the LFM waveform. In Fig.~\ref{fig:heatmaps}(\subref{fig:heatmaps_phasecoded}) we consider a discrete radar with a Zadoff-Chu (ZC) phase-coded rectangular waveform. Despite improved target discrimination capabilities, there are significant sidelobes due to the non-zero sidelobes of the waveform (cf. Fig.~\ref{fig:sa}(\subref{fig:sa_benedetto})). In contrast, our proposed approach generates a radar image with perfectly localized targets, as shown in Fig.~\ref{fig:heatmaps}(\subref{fig:heatmaps_zak}). 

\begin{figure}[!ht]
\centering
\begin{subfigure}{0.32\columnwidth}
    \includegraphics[width=\textwidth]{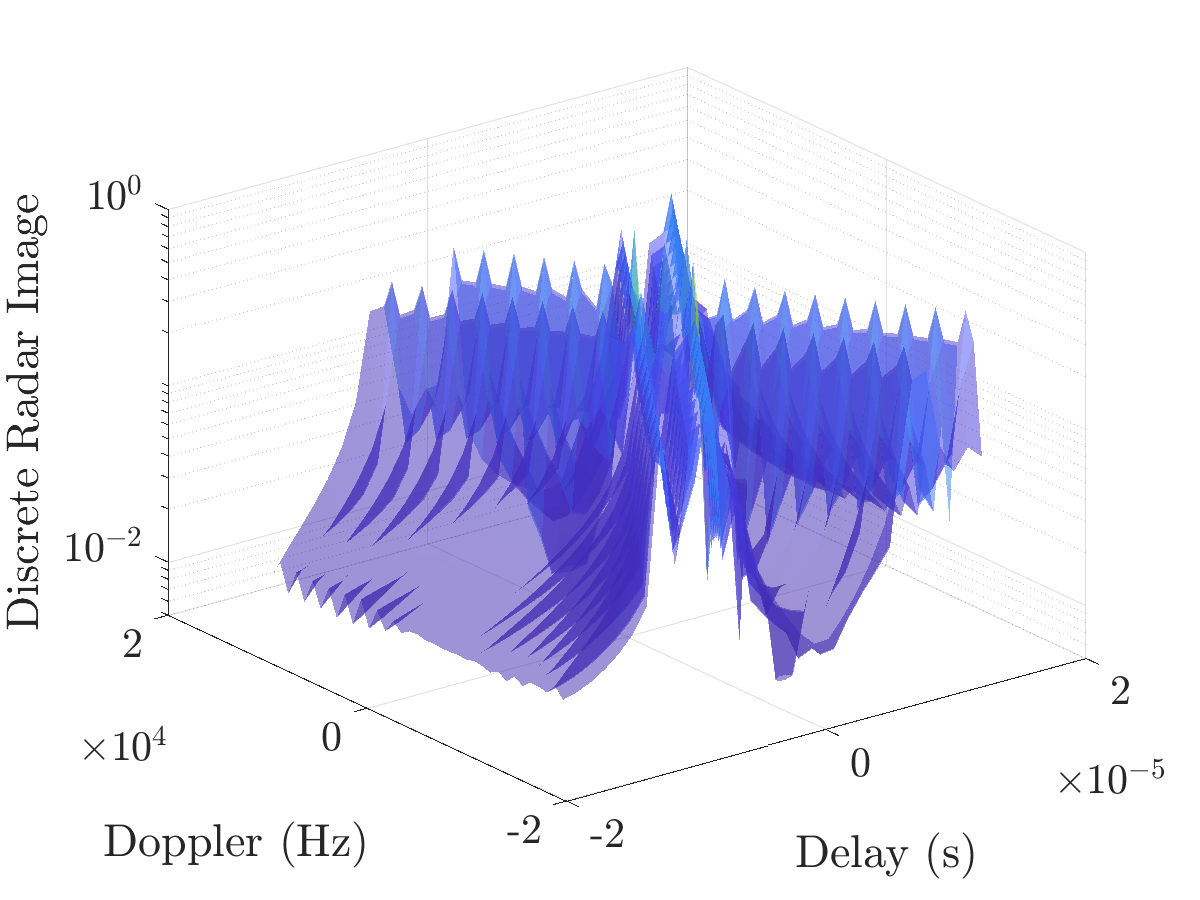}
\caption{Continuous (Fig.~\ref{fig:block_diag}(\subref{fig:cont_radar}))}
    \label{fig:heatmaps_fmcw}
\end{subfigure}
\begin{subfigure}{0.32\columnwidth}
    \includegraphics[width=\textwidth]{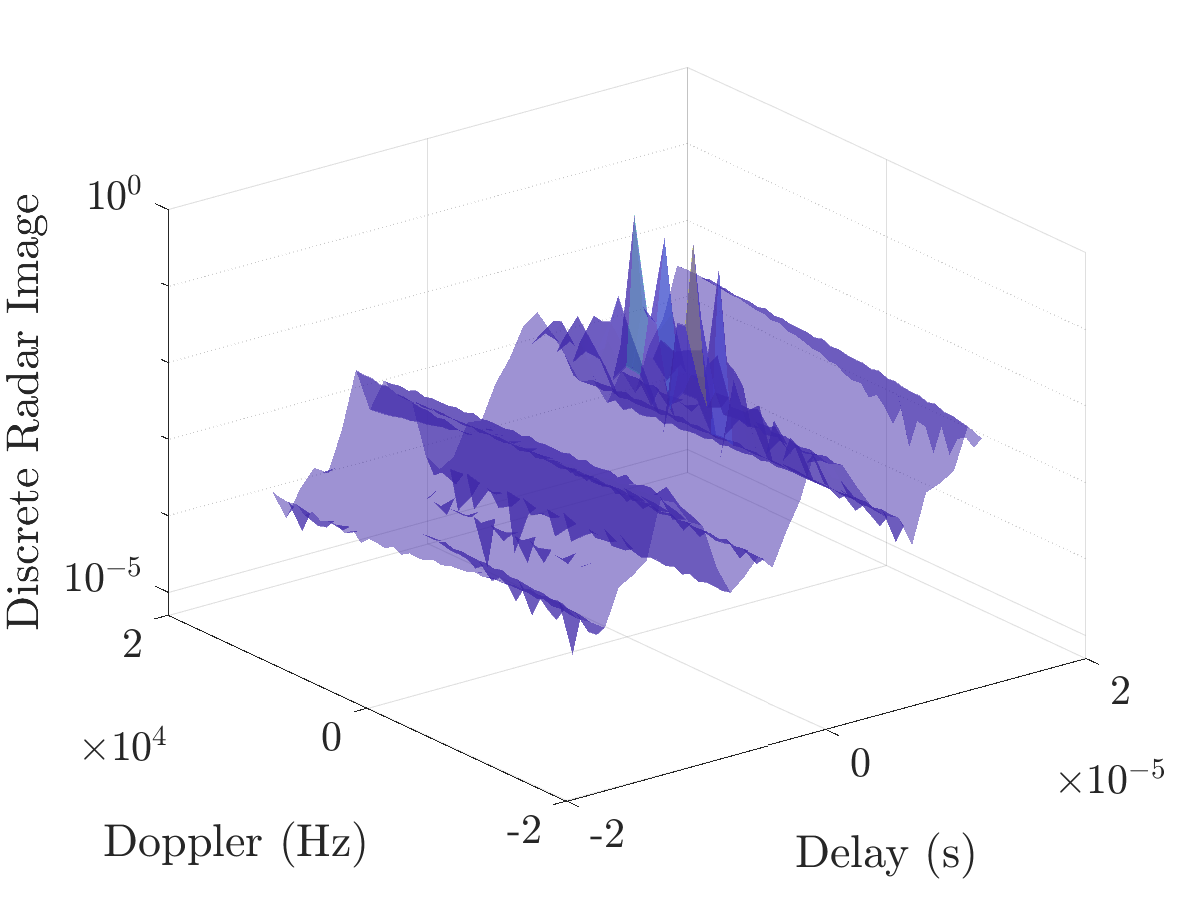}
\caption{Discrete (Fig.~\ref{fig:block_diag}(\subref{fig:disc_phase_coded_radar})).}
    \label{fig:heatmaps_phasecoded}
\end{subfigure}
\begin{subfigure}{0.32\columnwidth}
    \includegraphics[width=\textwidth]{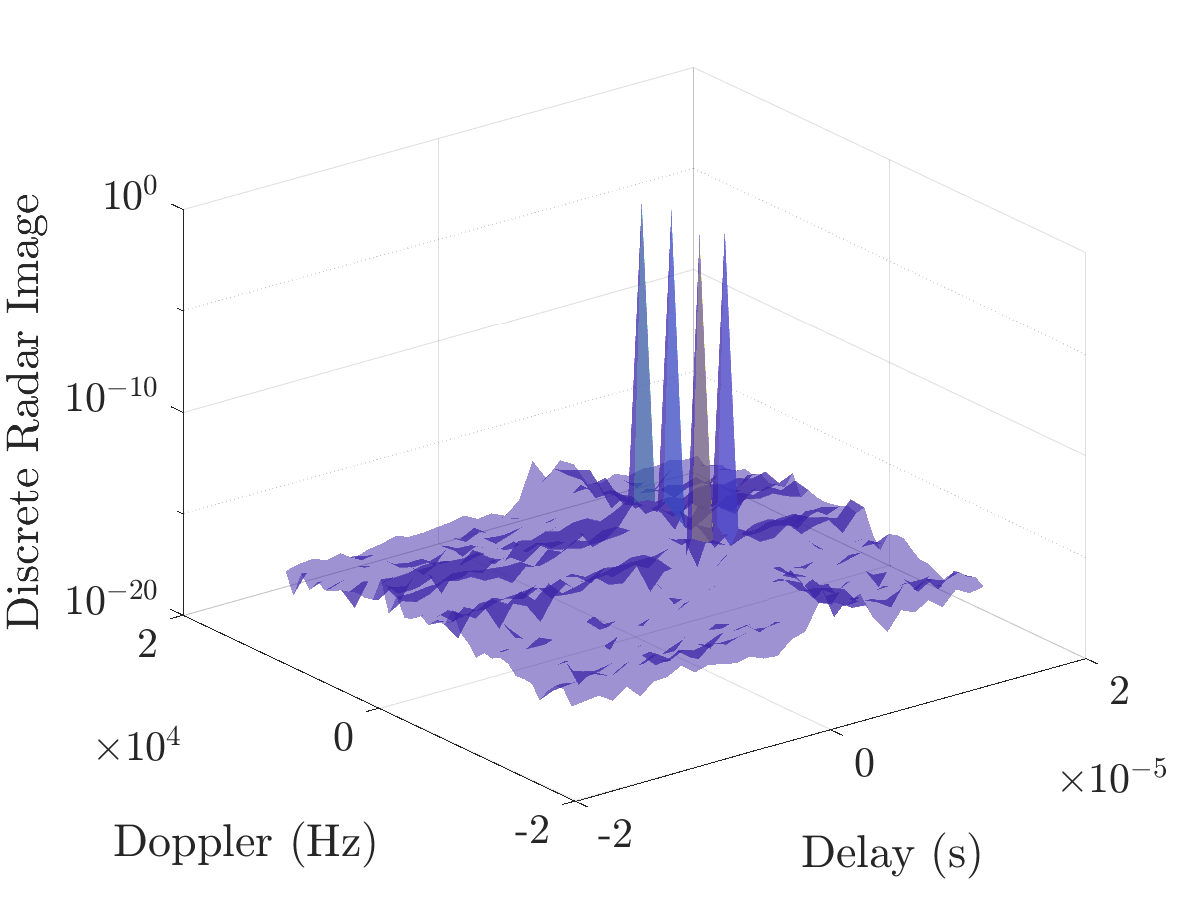}
\caption{Proposed (Fig.~\ref{fig:block_diag}(\subref{fig:disc_prop_radar})).}
    \label{fig:heatmaps_zak}
\end{subfigure}
\caption{The continuous and discrete radar architectures generate radar images with significant sidelobes around the ground truth target locations due to the choice of the carrier waveforms. Our proposed approach generates a perfectly localized radar image.}
\vspace{-5mm}
    \label{fig:heatmaps}
\end{figure}

\begin{remark}
    \label{rmk:radar1}
    Radars ask questions of scattering environments in order to control higher-layer functions such as target tracking. Images of scatterers are a means to an end rather than an end in itself. There may be small differences in the images of a scattering environment produced by continuous and discrete radars (cf. Fig.~\ref{fig:heatmaps}), but both images have the information necessary to control higher layer functions. We refer the reader to the work of Bell~\cite{Bell2002} for information theoretic measures of effectiveness, to Kershaw and Evans~\cite{Evans2002} for information theoretic criteria for waveform scheduling to support tracking. For more details, see~\cite{Moran2004_wvf_lib,Moran2009_wvf_lib_survey,Evans2002,Moran2006}. 
\end{remark}

\subsection{Benefit 1: ``Bed-of-Nails'' Ambiguity Functions}
\label{subsec:impl_radar_wvf}

One benefit of our proposed discrete radar architecture is that it enables perfectly localized \emph{``bed-of-nails''} ambiguity functions. Recall from Section~\ref{subsec:prelim_ambgfun} that an accurate estimate of the DD scattering environment is achieved when the self-ambiguity function of the transmit waveform is as ``thumbtack-like'' as possible, $\mathbf{A}_{\mathbf{x}}[k,l] \approx \delta[k] \delta[l]$. However, Moyal's Identity (Identity~\ref{idty:moyal} and Corollary~\ref{corr:moyal_setcard}) limits what can be achieved.

The standard approach taken by discrete radars (see Fig.~\ref{fig:block_diag}(\subref{fig:disc_phase_coded_radar})) is to regard the waveform as a signal modulated onto a carrier and to separate the carrier modulation and demodulation processes from the analysis of ambiguity. However, their performance is limited by the ambiguity function characteristics of the chosen carrier waveform. As an illustration, Fig.~\ref{fig:sa}(\subref{fig:sa_benedetto}) plots the magnitude of the self-ambiguity function for a Zadoff-Chu phase-coded rectangular waveform as defined in~\cite{benedetto_phasecoded}, with the waveform designed to have self-ambiguity function magnitude close to $1$ on the line $2\alpha k - l \equiv 0 \bmod{MN}$. The ambiguity characteristics of the rectangular carrier waveform results in significant sidelobes outside the locations given by the line.



\begin{figure}[!ht]
\centering
\begin{subfigure}{0.45\columnwidth}
    \includegraphics[width=\textwidth]{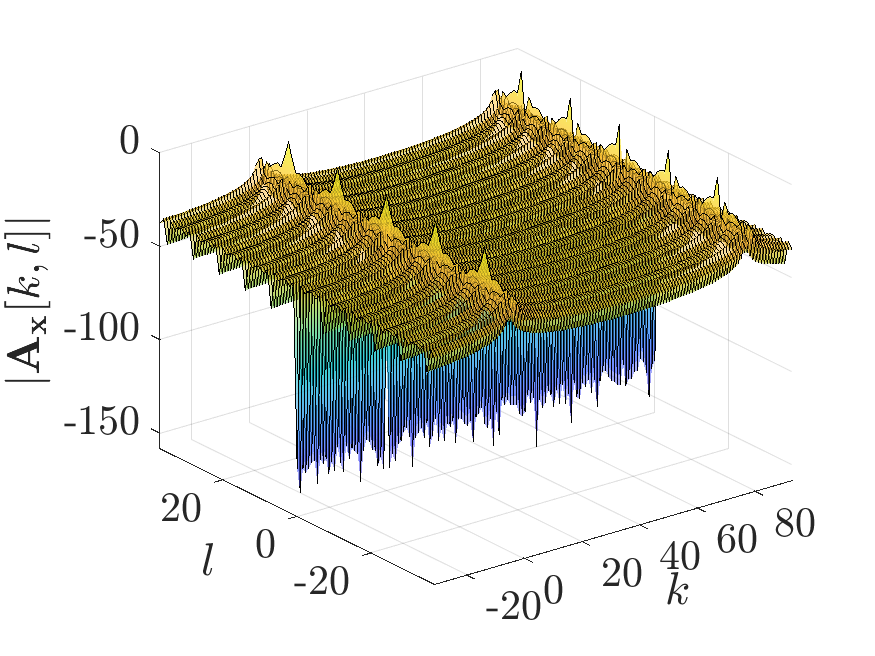}
\caption{ZC phase-coded waveform~\cite{benedetto_phasecoded}.}
    \label{fig:sa_benedetto}
\end{subfigure}
\begin{subfigure}{0.45\columnwidth}
    \includegraphics[width=\textwidth]{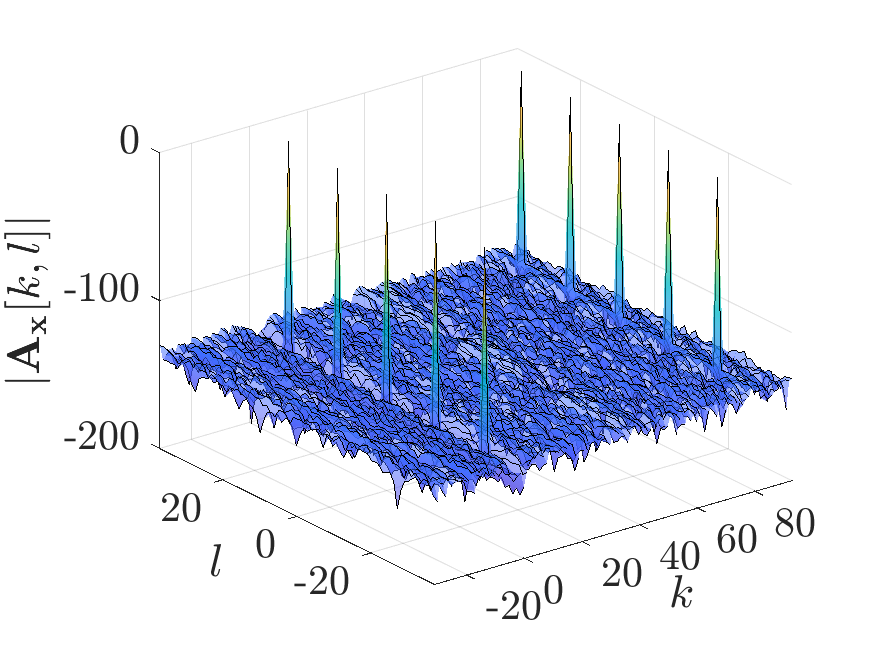}
\caption{Eigenvector in Lemma~\ref{lmm:eigenvec_amb_ex2_tdcazac}.}
    \label{fig:sa_proposed}
\end{subfigure}
\caption{The proposed discrete radar architecture in Fig.~\ref{fig:block_diag}(\subref{fig:disc_prop_radar}) achieves perfectly localized ``bed-of-nails'' ambiguity functions. (a) The standard approach of phase coding a rectangular waveform~\cite{benedetto_phasecoded} in discrete radars (see Fig.~\ref{fig:block_diag}(\subref{fig:disc_phase_coded_radar})) has high sidelobes due to poor ambiguity characteristics of the chosen carrier waveform. (b) Choosing the transmitted waveform as an eigenvector of a maximal commutative subgroup results in ``bed-of-nails'' ambiguity functions with minimal sidelobes.}
\vspace{-5mm}
    \label{fig:sa}
\end{figure}

In contrast, our proposed discrete radar architecture in Fig.~\ref{fig:block_diag}(\subref{fig:disc_prop_radar}) achieves \emph{perfectly localized} ``bed-of-nails'' ambiguity functions with minimal sidelobes. Specifically, we choose the radar waveform as an eigenvector of a maximal commutative subgroup. Recall from Corollary~\ref{corr:max_comm_ambg} that for a maximal commutative subgroup $\mathcal{T}_{MN}$ with eigenbasis $\big\{\mathbf{v}_{i}\big\}_{i=1}^{MN}$, the self-ambiguity function of any eigenvector $\mathbf{v}_{i}$ vanishes for all $e^{\frac{j2\pi}{MN}m} \mathcal{D}_{(k,l)} \not\in \mathcal{T}_{MN}$ and is unimodular only at $e^{\frac{j2\pi}{MN}m} \mathcal{D}_{(k,l)} \in \mathcal{T}_{MN}$. Fig.~\ref{fig:sa}(\subref{fig:sa_proposed}) illustrates the advantage of our approach. We consider the discrete chirp eigenvector from Lemma~\ref{lmm:eigenvec_amb_ex2_tdcazac} corresponding to the Heisenberg-Weyl maximal commutative subgroup example in Special Case~\ref{ex:comm_subgrp_ex2}, and observe essentially no sidelobes in the self-ambiguity function magnitude outside the self-ambiguity function support $2\alpha k - l \equiv 0 \bmod{MN}$.



\subsubsection{How to Select a Radar Waveform}
\label{subsubsec:radar_wvf_cryst}

How should one choose an appropriate maximal commutative subgroup $\mathcal{T}_{MN}$ in our approach? Moyal's Identity (Identity~\ref{idty:moyal} and Corollary~\ref{corr:moyal_setcard}) states that a perfect ``thumbtack-like'' ambiguity function is unachievable. However, it is feasible to design a waveform whose self-ambiguity function is zero at all locations excluding the origin in a connected region $\mathcal{C}$ corresponding to the maximum possible support of the scattering environment\footnote{e.g., $\mathcal{C} = [k_{\min},k_{\max}] \times [l_{\min},l_{\max}]$ based on prior knowledge of the minimum/maximum delay and Doppler spreads of the scattering environment.}, $(k,l) \in \mathcal{C}$, $(k,l) \neq (0,0)$. Recall from Corollary~\ref{corr:max_comm_ambg} that our approach of transmitting an eigenvector of a maximal commutative subgroup $\mathcal{T}_{MN}$ yields self-ambiguity functions that are unimodular only at delay-Doppler indices in the set $\mathcal{S}_{\mathcal{T}_{MN}} = \big\{(k,l) \big| e^{\frac{j2\pi}{MN}m} \mathcal{D}_{(k,l)} \in \mathcal{T}_{MN} \big\}$, akin to a ``bed-of-nails''. 

Therefore, we propose to choose $\mathcal{T}_{MN}$ such that translates of the maximum scattering environment support $\mathcal{C}$ by the elements of $\mathcal{S}_{\mathcal{T}_{MN}}$ do not overlap:
\begin{equation}
    \label{eq:no_aliasing}
    \bigg(\bigcup_{(k,l) \in \mathcal{S}_{\mathcal{T}_{MN}}} \big(\mathcal{C} + (k,l)\big) \bigg) \cap \bigg(\bigcup_{(k',l') \in \mathcal{S}_{\mathcal{T}_{MN}}} \big(\mathcal{C} + (k',l')\big) \bigg) = \emptyset,~(k,l) \neq (k',l').
\end{equation}

This is the \emph{crystallization condition} described in~\cite{bitspaper1,bitspaper2,otfs_book} and it implies that an image of the scattering environment can be read off from the response to the radar waveform at delay-Doppler locations within the scattering environment support $\mathcal{S}$. Fig.~\ref{fig:no_aliasing_ex} illustrates how to choose $\mathcal{T}_{MN}$ for a given support $\mathcal{C}$. 

\begin{figure}[!ht]
\centering
\begin{subfigure}{0.49\columnwidth}
    \includegraphics[width=\textwidth]{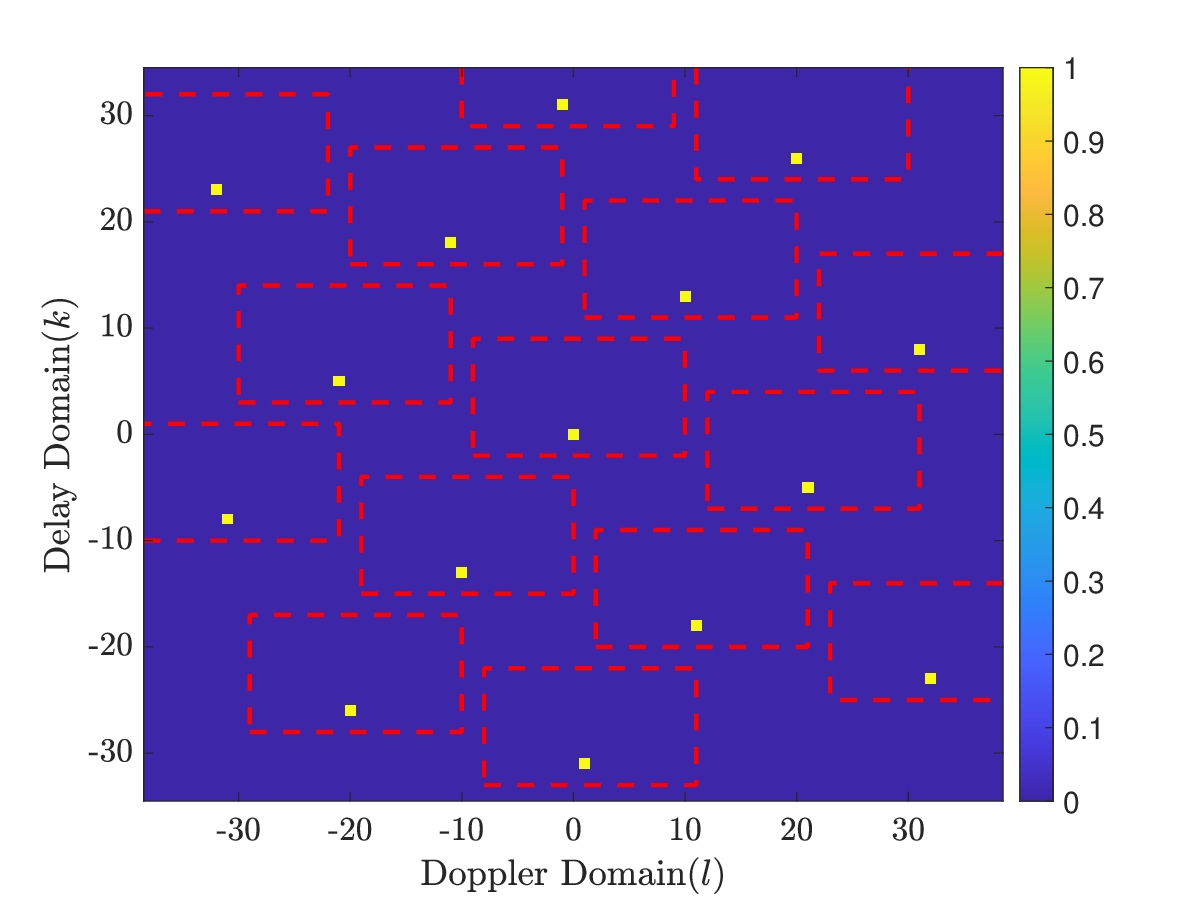}
\caption{Choice of $\mathcal{S}_{\mathcal{T}_{MN}}$ satisfying~\eqref{eq:no_aliasing}.}
    \label{fig:no_aliasing_ex1}
\end{subfigure}
\begin{subfigure}{0.49\columnwidth}
    \includegraphics[width=\textwidth]{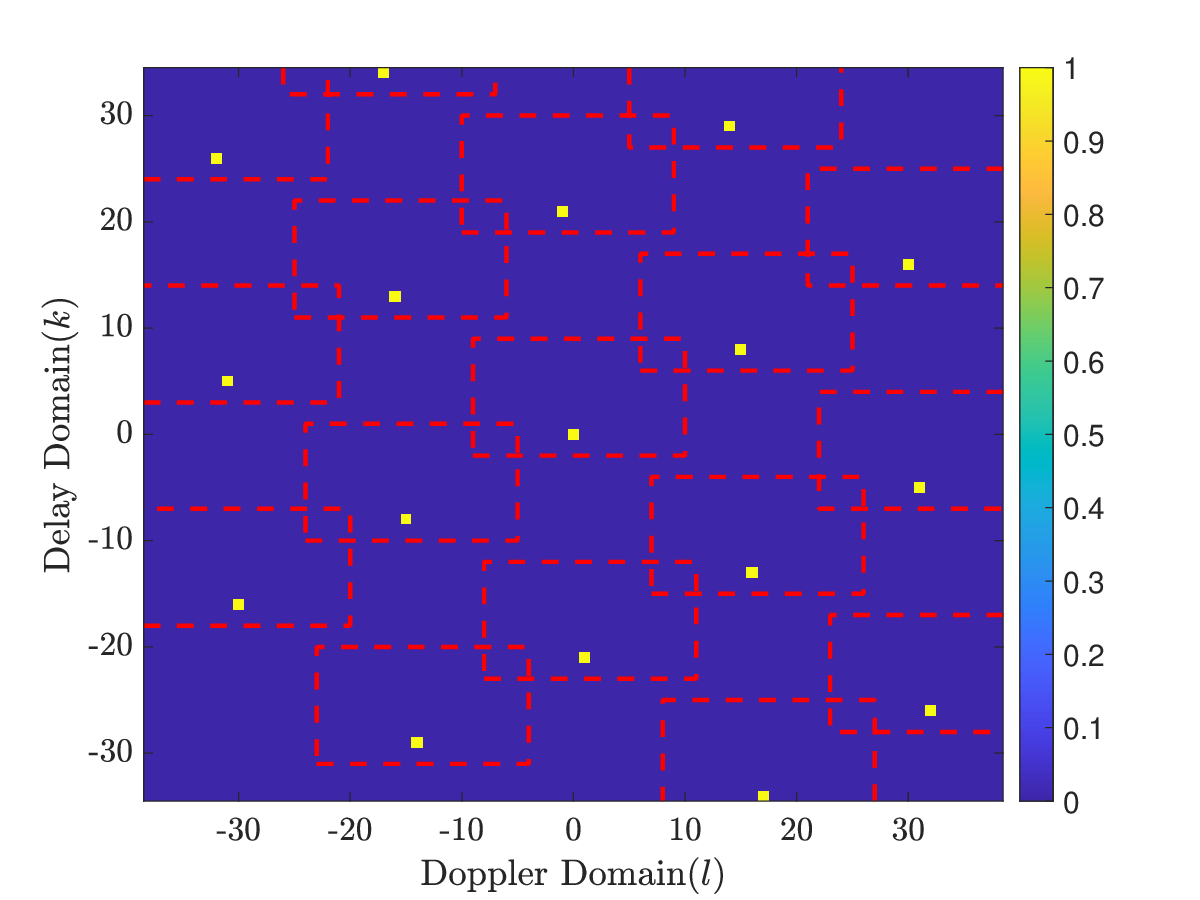}
\caption{Choice of $\mathcal{S}_{\mathcal{T}_{MN}}$ not satisfying~\eqref{eq:no_aliasing}.}
    \label{fig:no_aliasing_ex2}
\end{subfigure}
\caption{Example illustrating the choice of an appropriate maximal commutative subgroup $\mathcal{T}_{MN}$ whose delay-Doppler index set $\mathcal{S}_{\mathcal{T}_{MN}}$ satisfies the condition in~\eqref{eq:no_aliasing} for a given maximum scattering environment support $\mathcal{C}$. In this example, $\mathcal{C} = [k_{\min},k_{\max}] \times [l_{\min},l_{\max}]$ with $k_{\min} = -2$, $k_{\max} = 8$, $l_{\min} = -9$, $l_{\max} = 9$. Figure adapted from~\cite{Mehrotra2025_WCLSpread}.}
\vspace{-5mm}
    \label{fig:no_aliasing_ex}
\end{figure}

\subsection{Benefit 2: Waveform Libraries with Small PAPR}
\label{subsec:papr_radar}

Another benefit of our proposed discrete radar architecture is that it enables defining waveform libraries with low PAPR. Adaptive radars define libraries of waveforms from which an appropriate waveform is chosen in real-time based on the operational requirement~\cite{Moran2004_wvf_lib,Moran2009_wvf_lib_survey,Evans2002,Moran2006}. For instance, it was shown in~\cite{Moran2004_wvf_lib,Moran2009_wvf_lib_survey} that waveform libraries that rotate and chirp a template waveform maximize the mutual information for tracking applications. The ability to rotate a template waveform has also shown to be useful in sensing scattering environments where the product of the maximum delay and Doppler spreads exceeds the time-bandwidth product $BT$~\cite{Calderbank2025_isac}. At the same time, it is essential to ensure that each waveform in the library can be transmitted at low hardware cost. A standard metric to quantify the hardware cost of a waveform transmission is the PAPR~\cite{Jiang2008_papr}, which measures the ratio of the instantaneous power to the average power of a waveform. Waveforms with a large PAPR are undesirable from a hardware perspective since they require high dynamic range components and degrade the efficiency of power amplifiers.

The theory of symplectic transformations presented in Section~\ref{subsec:gdaft} enables meeting both objectives, i.e., defines waveform libraries with low PAPR. Recall from Lemma~\ref{lmm:weil_prop2} that symplectic transformations rotate the ambiguity functions of their inputs. Moreover, we presented examples of symplectic transformations in Section~\ref{subsec:gdaft}; the transformation in Example~\ref{ex:sympl_tx_ex1} linearly modulates the frequency of the input waveform, and the GDAFT in Example~\ref{ex:sympl_tx_ex2} rotates the self-ambiguity of the input waveform. The combination of the two transformations can be used to define a discrete radar waveform library following the principles presented in~\cite{Moran2004_wvf_lib,Moran2009_wvf_lib_survey}. Fig.~\ref{fig:wvf_lib} illustrates how the GDAFT rotates the self-ambiguity function of a pulsone basis element (Lemma~\ref{lmm:eigenvec_amb_ex1_pulsone}). 

\begin{figure}[t]
\centering
\begin{subfigure}{0.49\columnwidth}
    \includegraphics[width=\textwidth]{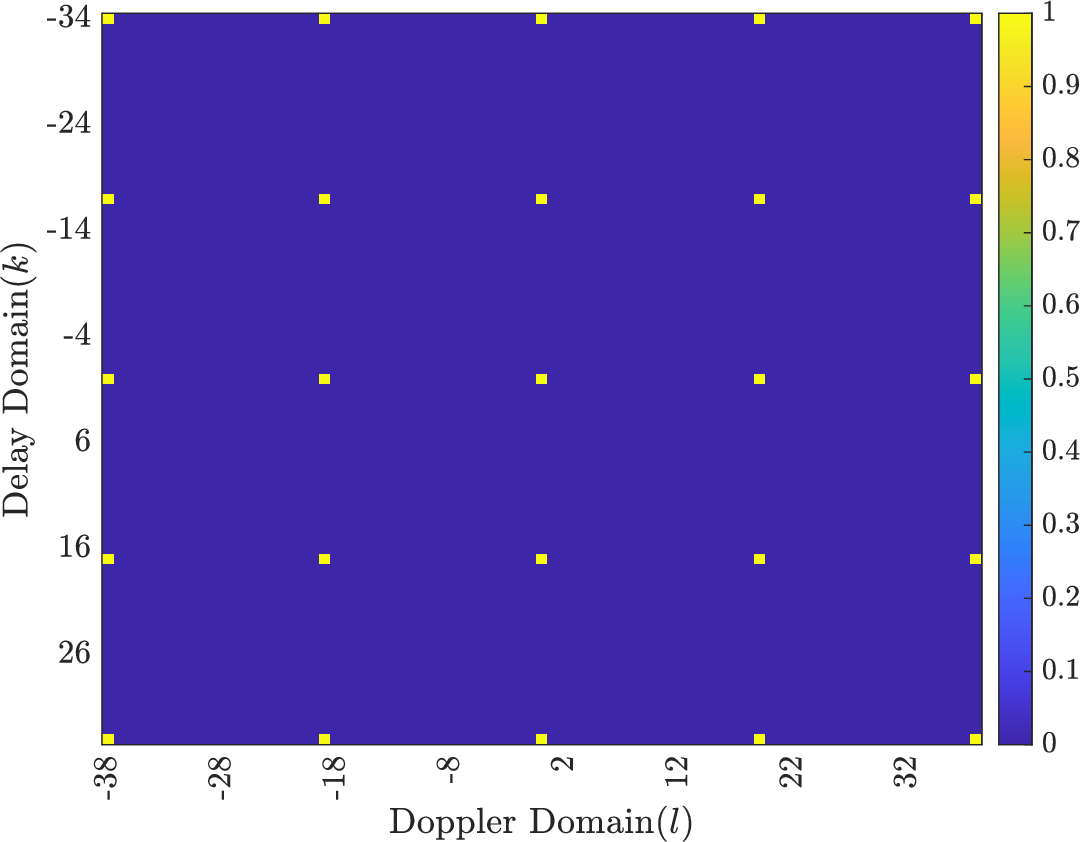}
\caption{Pulsone basis element.}
    \label{fig:wvf_lib1}
\end{subfigure}
\begin{subfigure}{0.49\columnwidth}
    \includegraphics[width=\textwidth]{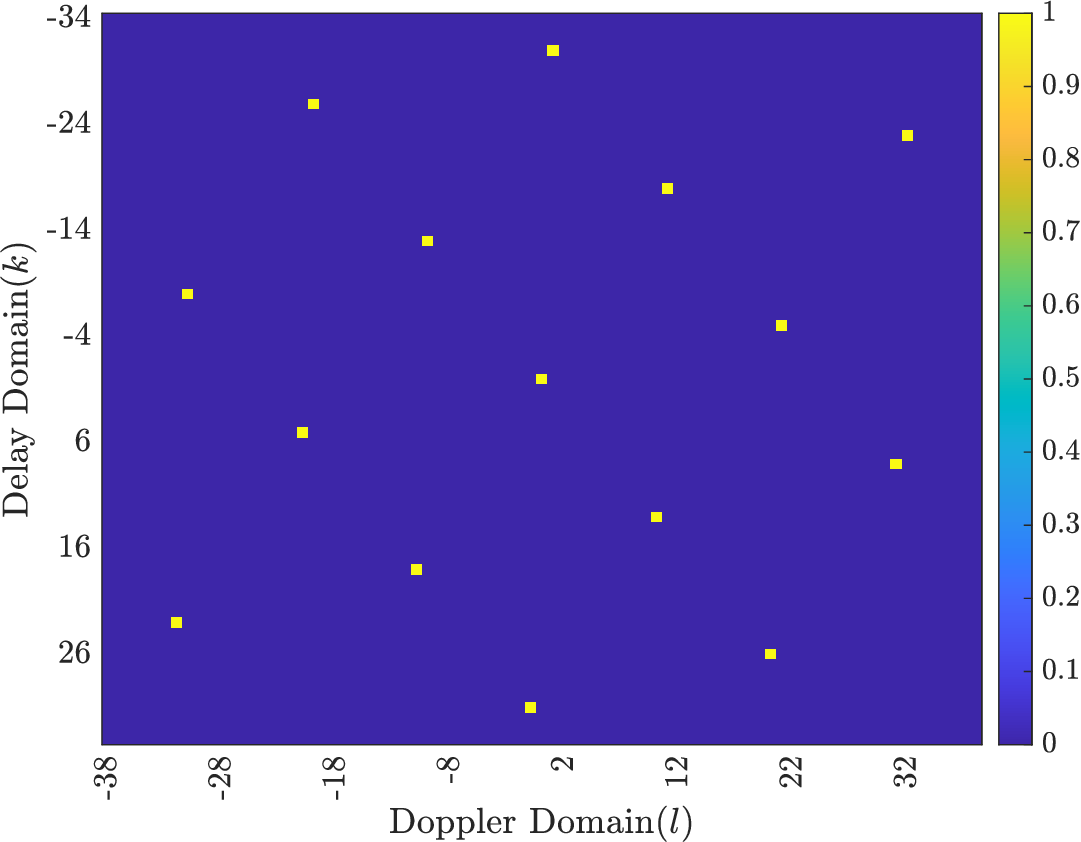}
\caption{GDAFT of pulsone.}
    \label{fig:wvf_lib2}
\end{subfigure}
\caption{The GDAFT (Example~\ref{ex:sympl_tx_ex2}) defines waveform libraries with rotated ambiguity functions of a template waveform, shown for a pulsone basis element (Lemma~\ref{lmm:eigenvec_amb_ex1_pulsone}).}
\vspace{-5mm}
    \label{fig:wvf_lib}
\end{figure}

A useful property of the GDAFT is that it maps pulsones to low PAPR waveforms, which have been previously used by the authors for spread carrier communication in~\cite{Mehrotra2025_WCLSpread}. Fig.~\ref{fig:papr} illustrates how the GDAFT maps the ``peaky'' pulsone waveform (Fig.~\ref{fig:papr}(\subref{fig:tdwvf_pulsone})) to a constant amplitude waveform (Fig.~\ref{fig:papr}(\subref{fig:tdwvf_cazac})). Fig.~\ref{fig:papr}(\subref{fig:tdwvf_papr}) plots the complementary cumulative distribution function (CCDF) of the PAPR (in dB) of both waveforms, showing a $5.6$ dB reduction in the PAPR on applying the GDAFT.

\begin{figure}[t]
\centering
\begin{subfigure}{0.32\columnwidth}
    \includegraphics[width=\textwidth]{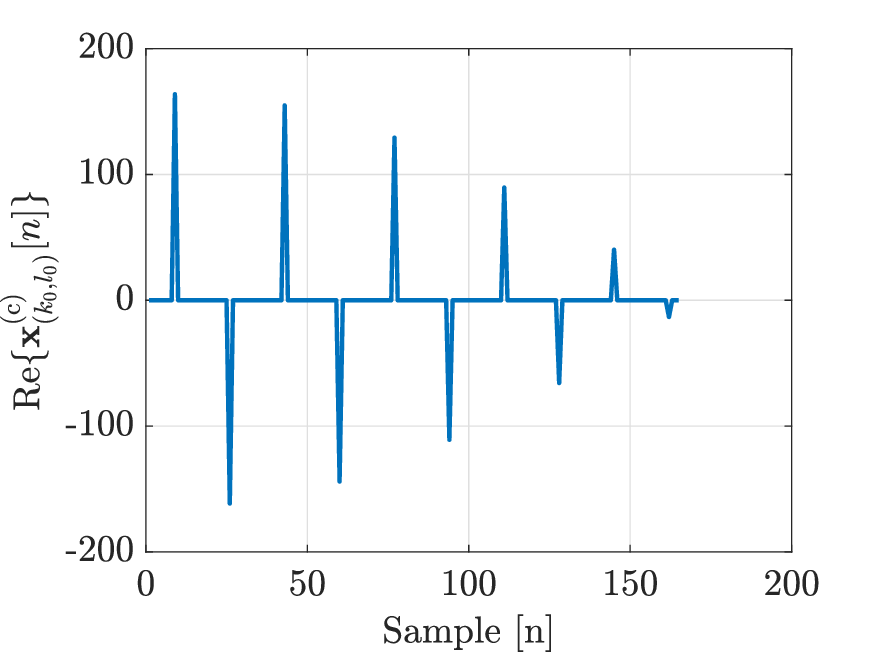}
\caption{Pulsone basis element.}
    \label{fig:tdwvf_pulsone}
\end{subfigure}
\begin{subfigure}{0.32\columnwidth}
    \includegraphics[width=\textwidth]{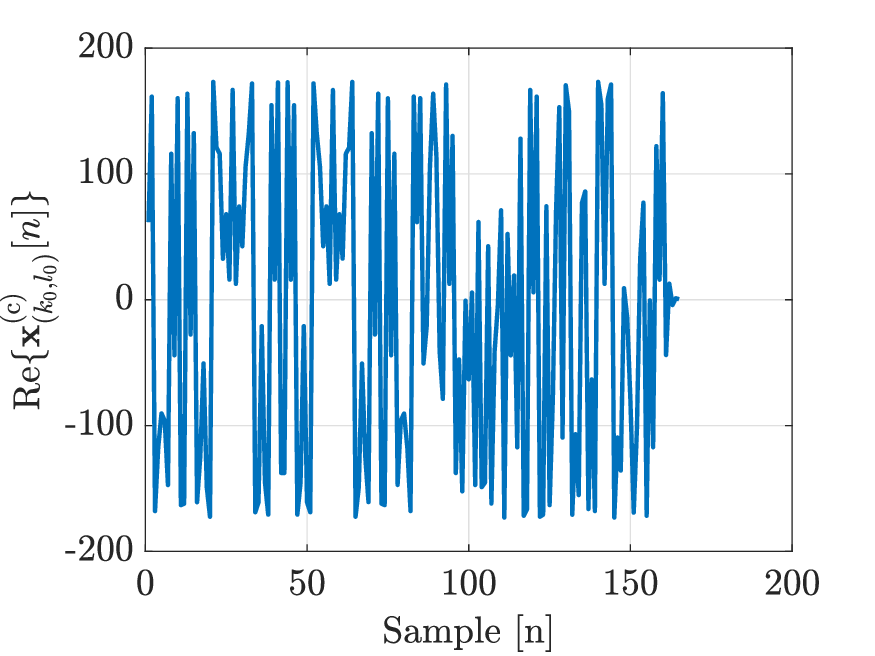}
\caption{GDAFT of pulsone.}
    \label{fig:tdwvf_cazac}
\end{subfigure}
\begin{subfigure}{0.32\columnwidth}
    \includegraphics[width=\textwidth]{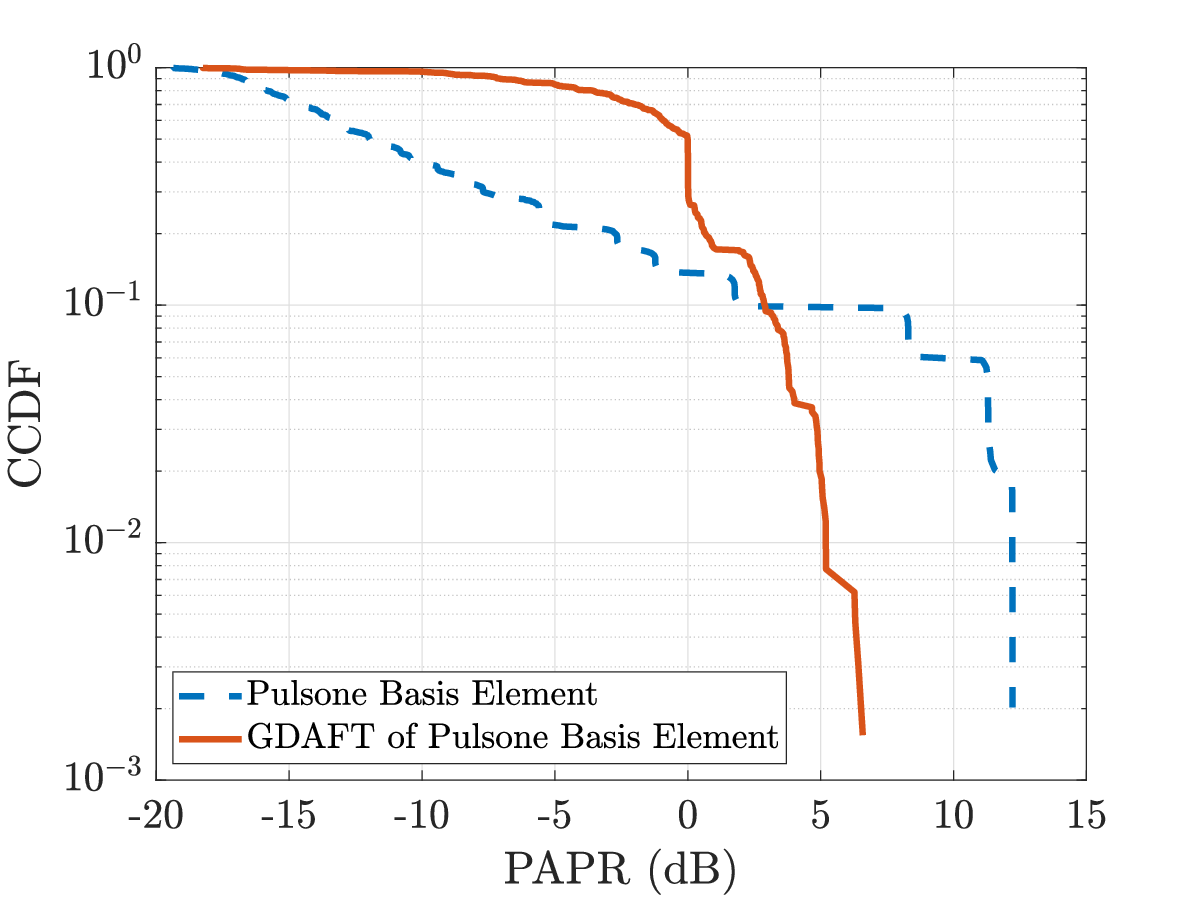}
\caption{PAPR comparison.}
    \label{fig:tdwvf_papr}
\end{subfigure}
\caption{The GDAFT (Example~\ref{ex:sympl_tx_ex2}) reduces the PAPR of the pulsone basis element (Lemma~\ref{lmm:eigenvec_amb_ex1_pulsone}) by about $5.6$ dB. Figure adapted from~\cite{Mehrotra2025_WCLSpread}.}
\vspace{-5mm}
    \label{fig:papr}
\end{figure}

\subsection{Benefit 3: Low-Complexity Cross-Ambiguity Computation}
\label{subsec:low_compl_radar}


The final benefit of our proposed discrete radar architecture is the complexity gain in the cross-ambiguity computation over the conventional continuous and discrete radar architectures. As described in Section~\ref{subsec:prelim_ambgfun}, radars form an image of the scattering environment by computing the cross-ambiguity function from Definition~\ref{def:amb_fun} between the received and transmitted signals. In the following Lemma, we show that the complexity of calculating the cross-ambiguity with our proposed discrete radar architecture in Fig.~\ref{fig:block_diag}(\subref{fig:disc_prop_radar}) is only $\mathcal{O}(BT\log T) = \mathcal{O}(MN\log N)$ when the pulsone basis element from Lemma~\ref{lmm:eigenvec_amb_ex1_pulsone} (or its symplectic transformation) is used as the transmit waveform. In contrast, the complexity of the cross-ambiguity with the continuous and discrete radar architectures in Figs.~\ref{fig:block_diag}(\subref{fig:cont_radar}) and (\subref{fig:disc_phase_coded_radar}) is $\mathcal{O}(B^2T^2) = \mathcal{O}(M^2N^2)$.


\begin{lemma}
    \label{lmm:low_compl_radar}
    The cross-ambiguity between two unit-norm $MN$-periodic sequences $\mathbf{x}$ and $\mathbf{y}$, where $\mathbf{y}$ is the Zak-OTFS pulsone from Lemma~\ref{lmm:eigenvec_amb_ex1_pulsone}, has $\mathcal{O}(MN\log N)$ complexity. Moreover, the complexity remains unchanged on applying a symplectic transformation to $\mathbf{x}$ and $\mathbf{y}$.
\end{lemma}

\begin{proof}
    Substituting the pulsone from Lemma~\ref{lmm:eigenvec_amb_ex1_pulsone} in Definition~\ref{def:amb_fun}:
    \begin{align}
        \label{eq:low_compl1}
        \mathbf{A}_{\mathbf{x},\mathbf{y}}[k, l] &= \sum_{n=0}^{MN-1}\mathbf{x}[n]\mathbf{y}^*[(n-k)_{{}_{MN}}]e^{-\frac{j2\pi}{MN}l(n-k)} \nonumber \\
        &= \frac{1}{\sqrt{N}}\sum_{n=0}^{MN-1}\mathbf{x}[n] e^{-\frac{j2\pi}{MN}l(n-k)} \sum_{d \in \mathbb{Z}} e^{-\frac{j2\pi}{N} d l_0} \delta[(n-k-k_0-dM)_{{}_{MN}}] \nonumber \\
        &= \frac{1}{\sqrt{N}}\sum_{\bar{k}=0}^{M-1}\sum_{\bar{l}=0}^{N-1}\mathbf{x}[\bar{k}+\bar{l}M] e^{-\frac{j2\pi}{MN}l(\bar{k}+\bar{l}M-k)} \sum_{d \in \mathbb{Z}} e^{-\frac{j2\pi}{N} d l_0} \nonumber \\ &~~~\times \delta[(\bar{k}+\bar{l}M-k-k_0-dM)_{{}_{MN}}] \nonumber \\
        &= \frac{1}{\sqrt{N}}\sum_{\bar{k}=0}^{M-1}\sum_{\bar{l}=0}^{N-1}\mathbf{x}[\bar{k}+\bar{l}M] e^{-\frac{j2\pi}{MN}l(\bar{k}+\bar{l}M-k)} \sum_{d \in \mathbb{Z}} e^{-\frac{j2\pi}{N} d l_0} \nonumber \\ &~~~\times \delta\bigg[\bar{k}-(k+k_0)_{{}_M}+\bigg(\bar{l}-d-\bigg\lfloor \frac{k+k_0}{M} \bigg\rfloor\bigg)M\bigg] \nonumber \\
        &= \frac{1}{\sqrt{N}}\sum_{\bar{k}=0}^{M-1}\sum_{\bar{l}=0}^{N-1}\mathbf{x}[\bar{k}+\bar{l}M] e^{-\frac{j2\pi}{MN}l(\bar{k}+\bar{l}M-k)} \sum_{d \in \mathbb{Z}} e^{-\frac{j2\pi}{N} d l_0} \mathds{1}\bigg\{\bar{k} \equiv (k+k_0)_{{}_M} \bigg\} \nonumber \\ &~~~\times \delta\bigg[d-\bar{l}+\bigg\lfloor \frac{k+k_0}{M} \bigg\rfloor\bigg] \nonumber \\
        &= \frac{1}{\sqrt{N}}\sum_{\bar{l}=0}^{N-1} \mathbf{x}[(k+k_0)_{{}_M}+\bar{l}M] e^{-\frac{j2\pi}{MN}l((k+k_0)_{{}_M}+\bar{l}M-k)} e^{-\frac{j2\pi}{N} (\bar{l}-\lfloor \frac{k+k_0}{M} \rfloor) l_0} \nonumber \\
        &= e^{-\frac{j2\pi}{MN} \big[l((k+k_0)_{{}_M}-k) - \lfloor \frac{k+k_0}{M} \rfloor l_0 M \big]} \bigg(\frac{1}{\sqrt{N}}\sum_{\bar{l}=0}^{N-1} \mathbf{x}[(k+k_0)_{{}_M}+\bar{l}M] e^{-\frac{j2\pi}{N}(l + l_0)\bar{l}}\bigg).
    \end{align}
    
    The summation within the parenthesis simply corresponds to an FFT operation, hence for a given $k \in \mathbb{Z}_{{}_{M}}$ its value for all $l \in \mathbb{Z}_{{}_{N}}$ can be computed in $\mathcal{O}(N\log N)$ complexity. Therefore, across all $(k,l) \in \mathbb{Z}_{{}_{M}} \times \mathbb{Z}_{{}_{N}}$, the complexity is $\mathcal{O}(MN\log N)$.

    A direct consequence of the above result in conjunction with Lemma~\ref{lmm:weil_prop2} is that the complexity remains $\mathcal{O}(MN\log N)$ on applying a symplectic transformation to $\mathbf{x}$ and $\mathbf{y}$. 
\end{proof}

\begin{lemma}[\cite{Fei2024_phasecoded_compl}]
    \label{lmm:compl_radarcoding}
    The cross-ambiguity between two unit-norm $MN$-periodic sequences $\mathbf{x}$ and $\mathbf{y}$, where $\mathbf{y}$ is an amplitude-, phase- or frequency-coded waveform, has $\mathcal{O}(M^2N^2)$ complexity.
\end{lemma}

\begin{proof}
    We substitute the general expression for an amplitude-, phase- or frequency-coded waveform~\cite{Pezeshki2009}:
    \begin{equation}
        \label{eq:radarcoding1}
        \mathbf{y}(t) = \sum_{m = 0}^{MN-1} \mathbf{z}_{m} s(t-m T_{\mathsf{c}}),
    \end{equation}
    where $\mathbf{z}$ denotes an $MN$-periodic coding sequence and $s(t)$ denotes the carrier waveform (``chip'') with chip duration $T_{\mathsf{c}}$. On sampling~\eqref{eq:radarcoding1} at sampling rate $\nicefrac{1}{T_{\mathsf{c}}}$, we obtain the $MN$-length sampled waveform:
    \begin{equation}
        \label{eq:radarcoding2}
        \mathbf{y}[n] = \sum_{m = 0}^{MN-1} \mathbf{z}_{m} s[n-m],
    \end{equation}
    where $\mathbf{y}[n] = \mathbf{y}(nT_{\mathsf{c}})$ and $s[n-m] = s(nT_{\mathsf{c}}-m T_{\mathsf{c}})$. It is assumed that the coding sequence $\mathbf{z}$ and the sampled carrier waveform are normalized appropriately such that $\mathbf{y}$ is unit-norm.

    Substituting~\eqref{eq:radarcoding2} in Definition~\ref{def:amb_fun}:
    \begin{align}
        \label{eq:radarcoding3}
        \mathbf{A}_{\mathbf{x},\mathbf{y}}[k, l] &= \sum_{n=0}^{MN-1}\mathbf{x}[n]\mathbf{y}^*[(n-k)_{{}_{MN}}]e^{-\frac{j2\pi}{MN}l(n-k)} \nonumber \\
        &= \sum_{n=0}^{MN-1}\mathbf{x}[n] \sum_{m=0}^{MN-1}\mathbf{z}^{*}_{m} s^{*}[(n-k)_{{}_{MN}}-m]e^{-\frac{j2\pi}{MN}l(n-k)} \nonumber \\
        &= \sum_{m=0}^{MN-1}\mathbf{z}^{*}_{m} \sum_{n=0}^{MN-1}\mathbf{x}[n] s^{*}[(n-k)_{{}_{MN}}-m]e^{-\frac{j2\pi}{MN}l(n-k)}.
    \end{align}

    In the absence of any additional structure on the coding sequence $\mathbf{z}$ and the sampled carrier waveform, the final expression in~\eqref{eq:radarcoding3} corresponds to $\mathcal{O}(M^2N^2)$ complexity.
\end{proof}

It has been shown in~\cite{Calderbank2015_ltv} that the complexity of calculating the cross-ambiguity with a continuous radar architecture based on LFM waveforms is also $\mathcal{O}(M^2N^2)$.

\section{Conclusion}
\label{sec:conclusion}
We have described an architecture for radar sensing where the radar waveform is formed in the DD domain, converted to the TD for transmission, match filtered on receive, then returned to the DD domain for ambiguity function analysis. The delay and Doppler periods that characterize the radar waveform are determined by the intended delay and Doppler resolution of the radar. We have described how the intended resolution determines a finite group of delay and Doppler shifts, which we have termed the discrete Heisenberg-Weyl group. We have used this discrete group to provide a framework for ambiguity function analysis, replacing the standard framework which is based on the continuous Heisenberg-Weyl group. By choosing the radar waveform to be a common eigenvector of a maximal commutative subgroup of the discrete Heisenberg-Weyl group, every delay and Doppler shift in the subgroup acts on the waveform by multiplication by a complex phase. Our choice of waveform leads to ambiguity functions that are constant on cosets of the maximal commutative subgroup. We have shown that this choice of waveform reduces the complexity of radar signal processing from quadratic to sublinear in the time-bandwidth product. We have also shown that our approach enables defining waveform libraries with small PAPR. 

\section{Declarations}

\subsection{Funding}

This work is supported by the National Science Foundation under grants 2342690 and 2148212, in part by funds from federal agency and industry partners as specified in the Resilient \& Intelligent NextG Systems (RINGS) program, by the Department of Science and Technology, Govt. of India under grant TPN-96226, and by the Air Force Office of Scientific Research under grants FA 8750-20-2-0504 and FA 9550-23-1-0249.

\subsection{Competing interests}

Not applicable.

\subsection{Ethics approval and consent to participate}

Not applicable.

\subsection{Consent for publication}

Not applicable.

\subsection{Data availability}

Not applicable.

\subsection{Materials availability}

Not applicable.

\subsection{Code availability}

Not applicable.

\subsection{Author contribution}

NM \& SRM contributed equally to the work and jointly derived the theoretical results and performed the simulations. SKM, RH \& RC conceived the idea for the manuscript. NM \& RC wrote the manuscript. All authors read and approved the final manuscript.

\bibliography{references}

\end{document}